\newif\ifanonym
\anonymtrue
\documentclass[11pt,letterpaper]{article}
\usepackage[margin=1in]{geometry}

\usepackage{amsmath,amsthm,amssymb,amsfonts}
\usepackage{xspace}
\usepackage{url}
\usepackage{hyperref}
\usepackage[dvipsnames]{xcolor}
\usepackage[ruled,linesnumbered,noend]{algorithm2e}
\usepackage{appendix}
\usepackage{babel}

\usepackage{cancel}

\usepackage{graphicx} % Required for inserting images
\usepackage{booktabs} % For formal tables
\usepackage[ruled]{algorithm2e} % For algorithms
\usepackage{subfigure}
\usepackage{caption}

\usepackage{multirow}
\usepackage{afterpage}
\usepackage{lipsum}
\usepackage{longtable}
\usepackage{booktabs}
\usepackage{thm-restate}
\usepackage{graphicx}
\usepackage{subcaption}
\usepackage[normalem]{ulem}
\usepackage{bm}
\usepackage{enumerate}
\usepackage{bbm}
\usepackage{mathtools}
\usepackage{mathrsfs}
\usepackage[capitalise,noabbrev]{cleveref}

\usepackage{comment}
\usepackage{tikz}
\usetikzlibrary{positioning, arrows.meta, fit, backgrounds, shapes, calc, decorations.pathreplacing}

\hypersetup{colorlinks=true,allcolors=blue}

\allowdisplaybreaks[4]

\makeatletter
\newcounter{HALG@line}
\renewcommand{\theHALG@line}{\thealgorithm.\arabic{ALG@line}}
\makeatother

\usepackage[obeyFinal,textsize=footnotesize]{todonotes}

\let\epsilon\varepsilon
\let\varnothing\emptyset

\theoremstyle{plain}
\newtheorem{theorem}{Theorem}[section]
\newtheorem{lemma}[theorem]{Lemma}

\newtheorem{corollary}[theorem]{Corollary}

\newtheorem{proposition}[theorem]{Proposition}

\theoremstyle{definition}
\newtheorem{definition}[theorem]{Definition}

\theoremstyle{remark}

\newtheorem*{remark*}{Remark}

\newcommand{\SJRE}{{\tt SJR-E}\xspace}
\newcommand{\lamSJR}{{\tt $\lambda$-SJR}\xspace}
\newcommand{\lamSJRExist}{{\tt \lamSJR-E}\xspace}
\newcommand{\lamSJRB}{{\tt $\lambda$-SJRB}\xspace}
\newcommand{\lamSJRBExist}{{\tt \lamSJRB-E}\xspace}
\newcommand{\AJRE}{{\tt AJR-E}\xspace}
\newcommand{\lamAJR}{{\tt $\lambda$-AJR}\xspace}
\newcommand{\lamAJRExist}{{\tt \lamAJR-E}\xspace}
\newcommand{\lamAJRB}{{\tt $\lambda$-AJRB}\xspace}
\newcommand{\lamAJRBExist}{{\tt \lamAJRB-E}\xspace}
\newcommand{\ThetaTwo}{\Theta_2^p}
\newcommand{\SigmaTwo}{\Sigma_2^p}
\newcommand{\NP}{\text{NP}}
\newcommand{\Poly}{\text{P}}

\newcommand{\vcmember}{{\tt VertexCoverMember}\xspace}
\newcommand{\sat}{{\tt $\neg\forall\exists$3SAT}\xspace}
\newcommand{\karpreduct}{{$\;\leq^p_m\;$}}

\begin{document}

\title{Computational Complexity of Strong and Average Justified Representation}


\author{%
  Yizhou Ai \\
  University of Toronto \\
  \texttt{yizhou.ai@mail.utoronto.ca} \\
  \and
  Biaoshuai Tao \\
  Shanghai Jiao Tong University \\
  \texttt{bstao@sjtu.edu.cn} \\
}
\date{}

%\textbf{}%
\maketitle
\begin{abstract}
    We study the approval-based multiwinner election problem where a set of $n$ voters cast approval-based ballots to a set of $m$ candidates, and we are to select a winner committee consisting of $k$ candidates. We consider two axioms: strong justified representation (SJR) and average justified representation (AJR). A winner committee satisfies SJR if the satisfaction for each voter in every $\ell$-cohesive group is at least $\ell$. AJR is a weaker axiom that requires the average satisfaction for each $\ell$-cohesive group to be at least $\ell$. It is well known that a winner committee satisfying AJR may not exist (and neither does SJR). In this paper, we study the computational complexity of the following decision problem: given an approval-based multiwinner election instance, decide if there exists a winner committee satisfying SJR/AJR. We prove that this problem is $\Theta_2^p$-complete for SJR, and $\Sigma_2^p$-complete for AJR.

    As byproducts, we derive some results that are interesting in their own right. Firstly, we show that adding one more adaptive query to an NP oracle on top of polynomially many non-adaptive NP queries does not add more computational power, and the resulting complexity class is still $\Theta_2^p$. Secondly, we construct a set system that can be useful in other applications, especially when doing reductions from typical satisfiability problems such as 3SAT. 
\end{abstract}

\section{Introduction}
In a \emph{multiwinner election}, $n$ voters jointly decide a winning committee of $k$ candidates from a set of $m$ candidates.
In an \emph{approval-based multiwinner election}, each voter's ballot has the form of a binary string of length $m$ indicating whether she approves each candidate.
Approval-based multiwinner elections have many applications, especially in political elections in many countries.

When we are to select only one winner (i.e., with $k=1$), the natural choice is the candidate with the most approval votes.
In the multi-winner setting, however, there are many different natural ways to decide the winning committee.
Simply selecting the $k$ candidates with the largest numbers of approvals can result in unfairness in many scenarios.
For example, suppose we have 150 voters forming two parties, where Party A consists of 100 voters and Party B consists of 50 voters.
Suppose there are 10 candidates, where the first 5 candidates are aligned with Party A and the remaining 5 candidates are aligned with Party B.
Voters from Party A approve the first 5 candidates, and voters from Party B approve the remaining 5 candidates.
If we are to select $k=3$ winners, the choice based on largest numbers of approvals results in a winning committee consisting of only candidates aligned with Party A, which is unfair for Party B.
A more natural choice is to select the candidates that \emph{proportionally represent} the two parties.
Since Party A has twice as many voters as Party B, a more natural choice of $k=3$ candidates would be to choose 2 candidates from the first 5 candidates and 1 candidate from the last 5 candidates.
Generalizing this example, we can define the fairness notion of \emph{proportional representation} as follows: whenever the number of voters in a party is at least $\ell\cdot\frac{n}{k}$, this party is justified for $\ell$ seats in the winning committee.
In our previous example with $n=150$, Party A deserves 2 seats while Party B deserves 1 seat.
The minimum number of voters that justifies one seat, $\frac{n}{k}$, is referred to as the \emph{Hare quota}, first proposed by Alexander Hamilton for use in the United States congressional apportionment~\cite{pukelsheim2017quota}.

The real-world situation can be more complex: the voters are not pre-specified by parties.
In these scenarios, recent work has been focusing on \emph{justified representation}~\cite{aziz2017jr} that extends the notion of ``parties'' to ``cohesive groups'' reflecting voters with similar preferences.
An $\ell$-cohesive group consists of at least $\ell\cdot\frac{n}{k}$ voters that approve at least $\ell$ common candidates.
Following the idea of the Hare quota, each $\ell$-cohesive group deserves $\ell$ seats in the winning committee.
A natural requirement of a winning committee $W$ is to ensure every voter in each $\ell$-cohesive group approves at least $\ell$ candidates in $W$, i.e., the minimum satisfaction of an $\ell$-cohesive group member should be at least $\ell$.
This notion is called \emph{Strong Justified Representation} (SJR), proposed by Aziz et al.~\cite{aziz2017jr} and Brill et al.~\cite{brill2025individual}.\footnote{Aziz et al.~\cite{aziz2017jr} proposed the notion of \emph{semi-strong justified representation}, which is a weaker variant of SJR where the above-mentioned condition is only required to hold for $1$-cohesive groups. Brill et al.~\cite{brill2025individual} first proposed the notion exactly as described above. They call it \emph{individual representation}. We decide to call it \emph{strong justified representation} since we think it clearly belongs to the family of ``justified representation'' and it is the strongest notion in this chain (see \eqref{eqn:JRrelations} and Fig.~\ref{fig:jr} in Sect.~\ref{sec:related_work}).}

As its name suggests, this fairness notion is strong, and there exist election instances where an SJR committee does not exist.
Consider the example with 6 voters $v_1,\ldots,v_6$ and 3 candidates $c_1,c_2,c_3$ where candidate $c_1$ is approved by $v_1,v_2,v_3$, candidate $c_2$ is approved by $v_3,v_4,v_5$, and candidate $c_3$ is approved by $v_5,v_6,v_1$.
Suppose we want to select a committee of $k=2$ winners.
There are three $1$-cohesive groups: $\{v_1,v_2,v_3\}$, $\{v_3,v_4,v_5\}$, and $\{v_5,v_6,v_1\}$.
Since each voter is in at least one $1$-cohesive group, to achieve SJR, we need to select a committee $W$ of $2$ winners such that every voter approves at least one candidate in $W$.
This is impossible: voter $v_2$ only approves $c_1$, voter $v_4$ only approves $c_2$, voter $v_6$ only approves $c_3$, and we can only select 2 candidates.

Instead of requiring the \emph{minimum} satisfaction of every $\ell$-cohesive group member to be $\ell$, a natural weaker requirement is to enforce the \emph{average} satisfaction of all voters in every $\ell$-cohesive group to be at least $\ell$.
This notion is called \emph{Average Justified Representation} (AJR), introduced by Fern{\'{a}}ndez, Elkind, Lackner, Garc{\'{\i}}a, Arias{-}Fisteus, Basanta{-}Val, and Skowron~\cite{fernandez2017pjr}.\footnote{To be precise, although Fern{\'{a}}ndez et al.~\cite{fernandez2017pjr} discuss average satisfaction in cohesive groups (Sect.~5 in their paper), the name \emph{Average Justified Representation} was not used in Fern{\'{a}}ndez et al.~\cite{fernandez2017pjr}. We follow the convention of Han et al.~\cite{han2026likelihood} where this name was used.}
Notice that, however, an AJR committee is also not guaranteed to exist: the same example in the previous paragraph can show this.
Nevertheless, if we relax the average satisfaction requirement from $\ell$ to $\ell-1$, a desired committee is guaranteed to exist, and the well-known \emph{Proportional Approval Voting Rule} (the PAV rule) is guaranteed to output such a committee~\cite{aziz2018complexity,skowron2021proportionality}.

Motivated by the appealing properties of SJR and AJR, we would like to select a winner committee satisfying the SJR and/or AJR requirements if one exists.
\textbf{In this paper, we study the computational complexity of deciding the existence of an SJR/AJR committee in a given approval-based multiwinner election instance.}

The computational challenges come from two aspects.
Other than the challenge in finding a potential SJR/AJR committee we could typically expect, it is also challenging to \emph{verify} if a given committee $W$ satisfies the SJR/AJR axiom.
In fact, it is already challenging to see if a voter is in some $\ell$-cohesive group for large values of $\ell$.
Consider a bipartite graph where the vertices on two sides represent voters and candidates respectively, and the edges represent approvals.
Deciding the containment in an $\ell$-cohesive group becomes the problem of deciding the existence of an $\frac{\ell n}k\times \ell$ bi-clique in a bipartite graph, and the problem becomes NP-complete if $\ell$ is given as an input.
Aziz et al.~\cite{aziz2018complexity} show that it is coNP-complete to decide if a winning committee satisfies \emph{Extended Justified Representation} (a notion that is substantially weaker than AJR and SJR), and the proof exploits the hardness of locating $\ell$-cohesive groups.
The same argument can be used to show the coNP-completeness of AJR and SJR verification problems.
We include proofs for these in Appendix~\ref{sec:coNPhardness}.

On the other hand, the problem of deciding the existence of an SJR or an AJR committee is obviously in $\SigmaTwo$.
Given a committee $W$ and an $\ell$-cohesive group $N'$, we can check in polynomial time if the minimum (for SJR) or the average (for AJR) satisfaction of voters in $N'$ is at least $\ell$.
By definition, an instance is a yes instance if and only if \emph{there exists} a committee $W$ such that \emph{for all} $\ell$-cohesive groups $N'$ the minimum/average satisfaction of $N'$ is at least $\ell$.
This features the problem's membership in $\SigmaTwo$.
Other than this, the precise computational complexities of these two decision problems are yet to be understood.

\subsection{Our Results}
We study the two computational complexity problems: 1) given an election instance, decide if there is an AJR committee, and 2) given an election instance, decide if there is an SJR committee.
Our main results show that the first problem is $\SigmaTwo$-complete (Sect.~\ref{sec:Main Result in AJR}) and the second problem is $\ThetaTwo$-complete (Sect.~\ref{sec:Main Result}).
Our results indicate that the problem with AJR is substantially harder than that with SJR, although SJR is a stronger notion.

When proving our main results, we derive some intermediate results that are interesting on their own in theoretical computer science.

Firstly, we derive a complexity result $\Poly_{\parallel+1}^\NP=\ThetaTwo$ (Sect.~\ref{sec:complexityresult}).
To briefly explain this, $\ThetaTwo$ is the complexity class consisting of problems that can be solved in polynomial time by non-adaptively accessing an NP oracle for a polynomial number of times (see Sect.~\ref{sec:complexitybasics} if readers are unfamiliar with the complexity class $\ThetaTwo$).
We find that the above-mentioned problem with SJR is seemingly harder than problems in $\ThetaTwo$: it requires a polynomial number of non-adaptive NP queries \emph{plus one more adaptive NP query that depends on the outcome of previous non-adaptive queries}.
We denote the corresponding complexity class by $\Poly_{\parallel+1}^\NP$, and we show that this complexity class coincides with $\ThetaTwo$, i.e., the extra adaptive NP query at the end does not add more computational power.
This result is crucial for showing that the decision problem with SJR is still in $\ThetaTwo$.

Secondly, for proving our result with AJR (the $\SigmaTwo$-completeness result), we construct a set system that may be potentially useful in other problems (Sect.~\ref{sec:setsystem}).
The set system contains $2n$ sets that are partitioned into $n$ groups of two sets: $\mathcal{F}=\bigcup_{i=1}^n\mathcal{F}_i$ where $\mathcal{F}_i=\{F_i,\bar{F}_i\}$.
The property is that, if we want to choose $n$ sets (out of the $2n$ sets) with a maximum intersection, we must choose exactly one set from each group $\mathcal{F}_i$.
A na\"{\i}ve construction requires $2^n$ elements (see the second paragraph in Sect.~\ref{sec:setsystem}).
We design a more careful construction with $O(n^2)$ elements such that the set system becomes a useful gadget in reductions by not blowing up the size of the instance too much.
Our set system can be useful in many other reduction applications, especially, when reducing from typical satisfiability problems (such as 3SAT), the choice of $n$ sets can represent a truth assignment to the $n$ variables, and our construction distinguishes the choices representing truth assignments from the others by maximizing the set intersection.

From an algorithmic perspective, our results show that the SJR decision problem can be solved using logarithmically many adaptive NP-oracle calls.
The AJR decision problem, on the other hand, remains substantially harder in the polynomial hierarchy.

\subsection{Related Work}
\label{sec:related_work}

\paragraph{Related work in SJR and AJR.}
Most relevant to this paper, Brill et al.~\cite{brill2025individual} established that determining whether an SJR committee (referred to as IR in their work) exists is NP-hard. 
Our $\Theta_2^p$-completeness result strengthens this finding by pinpointing the exact computational complexity of the existence problem.

The axioms of Strong Justified Representation (SJR) and Average Justified Representation (AJR) have also been extensively examined through a probabilistic lens. 
Since an SJR or AJR committee is not guaranteed to exist, a significant body of literature has investigated the probability of their existence when election instances are generated via probabilistic models. 
This likelihood problem is well-characterized both empirically~\cite{bredereck2019experimental,PetersP0021,SzufaFJLSST22,elkind2023justifying,brill2025individual} and theoretically~\cite{Xia2025linear,han2026likelihood}.

\paragraph{On other JR notions.}
Aziz et al.~\cite{aziz2017jr} introduced the axiom of \emph{Extended Justified Representation} (EJR), which stipulates that in every $\ell$-cohesive group, there exists at least one voter with a satisfaction level of at least $\ell$. 
A strictly weaker notion, \emph{Proportional Justified Representation} (PJR), was proposed in Reference~\cite{fernandez2017pjr} and further analyzed in Reference~\cite{aziz2018complexity, brill2024phragmen}. 
PJR mandates that every $\ell$-cohesive group $N'$ satisfies
\[
    \left|W\cap\bigcup_{i\in N'}A_i\right|\ge \ell,
\]
whereas AJR requires the average satisfaction of $N'$ to be at least $\ell$. 
The least restrictive notion is \emph{Justified Representation} (JR)~\cite{aziz2017jr}, which applies the EJR condition exclusively to $1$-cohesive groups.

The hierarchical relationship among these standard notions~\cite{aziz2017jr,fernandez2017pjr,brill2025individual} is given by:
\begin{equation}\label{eqn:JRrelations}
    \text{SJR} \Rightarrow \text{AJR} \Rightarrow \text{EJR} \Rightarrow \text{PJR} \Rightarrow \text{JR}.
\end{equation}

Crucially, unlike SJR and AJR, committees satisfying EJR (and by implication, PJR and JR) are guaranteed to exist~\cite{aziz2017jr}. 
Consequently, the existence problem—central to our analysis of SJR and AJR—is rendered trivial for these weaker axioms.

Alternative variants have also been proposed, including FJR~\cite{peters2020proportional}, FPJR~\cite{KalayciLK2025}, EJR+, and PJR+~\cite{brill2023ejr+}. 
Furthermore, prior literature has framed the problem within an \emph{optimization} setting by considering \emph{quantitative metrics} of fairness, such as the \emph{proportionality degree}~\cite{aziz2018complexity, skowron2021proportionality} and \emph{EJR degree}~\cite{TaoZZ24}. 

Another line of research has been focusing on \emph{core stability}~\cite{gillies1959solutions,shapley1969market} which has also been studied widely~\cite{JiangMW20,MunagalaSWW22,MavrovMS23,Xia2025linear}.
In the context of multiwinner elections, a winning committee $W$ is core stable if there does not exist another committee $W'$ with $|W'|=\ell$ such that each voter in $N'$ has a strictly larger satisfaction for $W'$ than for $W$ for some voter set $N'$ with $|N'|\geq\ell\cdot\frac nk$.
The existence of a core-stable committee was a longstanding open problem and has recently been established by Becker, Greger, and Peters~\cite{becker2026existence}, who also give a polynomial-time algorithm for finding such a committee.
It is straightforward to check that both core stability and AJR imply EJR.
Core stability is incomparable with AJR or SJR.
AJR and SJR place no requirement on voter sets $N'$ with $|N'|\geq\ell\cdot\frac nk$ that are not $\ell$-cohesive, while core stability still requires that no alternative choice $W'$ with $|W'|=\ell$ exists such that everyone in $N'$ is strictly happier.
On the other hand, for an $\ell$-cohesive group $N'$, the requirement of AJR (and so of SJR) is stronger: if the average satisfaction of $N'$ is at least $\ell$, there is no way we can find the alternative $W'$ of $\ell$ candidates postulated by core stability.
The relationship between these notions is given in Fig.~\ref{fig:jr}.

\begin{figure}[!htb]
                \centering
                \begin{tikzpicture}[
                scale=0.8,
                transform shape,
                node distance=0.5cm and 2cm,
                every node/.style={font=\bfseries},
                ->, >=stealth]
                \definecolor{fancyblue}{RGB}{127,172,204}    
                \definecolor{fancyyellow}{RGB}{246,189,78}        
                \definecolor{fancyred}{RGB}{233,108,102}    
                \definecolor{fancygreen}{RGB}{31,145,158}        
                        % Nodes
                        \node (CS) {Core Stability};
                        \node[right=of CS] (FJR) {FJR};
                        \node[above=of FJR] (EJR) {EJR};
                        \node[right=of FJR] (FPJR) {FPJR};
                        \node[above=of FPJR] (PJR) {PJR};
                        \node[right=of PJR] (JR) {JR};
                        \node[left=of EJR] (AJR) {AJR};
                        \node[left=of AJR] (SJR) {SJR};
                        %\node[left=of AJR] (IR) {IR};
                        \node[above=of EJR] (EJRplus) {EJR+};
                        \node[above=of PJR] (PJRplus) {PJR+};
                        % Arrows
                        \draw (SJR) -> (AJR);
                        \draw (CS) -> (FJR);
                        \draw (FJR) -> (FPJR);
                        \draw (FJR) -> (EJR);
                        \draw (FPJR) -> (PJR);
                        \draw (EJR) -> (PJR);
                        \draw (PJR) -> (JR);
                        \draw (EJRplus) -> (EJR);
                        \draw (EJRplus) -> (PJRplus);
                        \draw (PJRplus) -> (PJR);
                        \draw (AJR) -> (EJR);
                        %\draw (IR) -> (AJR);
                \end{tikzpicture}
                \caption{Diagram of relationships among different variants of JR. Arrows indicate implications.}
                \label{fig:jr}
\end{figure}

For a comprehensive overview of related work, we refer the reader to the book~\cite{lackner2023multiwinner}.

\section{Preliminaries}
In the following, we recall the formal model of approval-based elections and introduce the key concepts needed to state our main problem.
%First, we fix notation for approval-based elections.

We begin by defining approval-based elections and committees.
\begin{definition}[Approval-Based Election]
    An \emph{approval-based election} is a tuple
    \[
        \mathbb{E} = (N, C, \mathcal{A}, k)
    \]
    where:
    \begin{itemize}
        \item $N=[n]=\{1,2,\dots,n\}$ is the set of \emph{voters},
        \item $C$ is a finite set of \emph{candidates},
        \item $\mathcal{A}=(A_1,\dots,A_n)$ is the list of \emph{approval ballots}, with each $A_i\subseteq C$ the set of candidates approved by voter $i$,
        \item $k\in\mathbb{Z}^+$ is the \emph{committee size}, satisfying $1\le k\le|C|$.
    \end{itemize}
\end{definition}

%Next, we formalize the notion of a committee for a given election.
\begin{definition}[Committee]
    Given an election $\mathbb{E}=(N,C,\mathcal{A},k)$, a \emph{committee} is any subset $W\subseteq C$ of size $|W|=k$.
\end{definition}

%\paragraph{$\ell$-Cohesive Groups, SJR, and AJR.}
The following notion captures the idea of sufficiently large and like-minded coalitions.
\begin{definition}[$\ell$-Cohesive Group]
    Let $\mathbb{E}=(N,C,\mathcal{A},k)$ and $1\le \ell\le k$. A subset of voters $N'\subseteq N$ is \emph{$\ell$-cohesive} if:
    \[
        |N'|\;\ge\;\frac{\ell\cdot n}{k}
        \quad\text{and}\quad
        \Bigl|\bigcap_{i\in N'}A_i\Bigr|\;\ge\;\ell.
    \]
    The first condition ensures the \emph{group-size requirement}, while the second enforces the \emph{common-approval requirement}.
\end{definition}

Besides, we generalize the above definition to $(\lambda,\ell)$-cohesive groups as follows. The notion of $(\lambda,\ell)$-group is mainly used for our technical proofs. Notice that an $\ell$-cohesive group is a $(\lambda,\ell)$-group with $\lambda=\frac{n}{k}$.

\begin{definition}[$(\lambda,\ell)$-group]\label{def:lambdaellgroup}
 Let $\lambda>0$ and let $\ell$ be a positive integer.
 A group of voters $N'\subseteq N$ is a $(\lambda,\ell)$-group if $|N'|\ge \ell\lambda$ and $|\bigcap_{i\in N'}A_i|\ge \ell$.
\end{definition}

A $(\lambda,\ell)$-group can exist only if
\[
    1\le \ell\le
    \min\left\{\left\lfloor\frac{|N|}{\lambda}\right\rfloor,|C|\right\}.
\]

As a remark, the parameter $\lambda$ can be viewed as specifying the quota used in the cohesiveness requirement.
The standard Hare quota version of justified representation corresponds to $\lambda=n/k$, while other quota choices, such as \emph{the Droop quota}~\cite{droop1881methods} $q_D\approx n/(k+1)$, have also been studied in the literature.
Quota-scaling ideas also appear in approximate core stability~\cite{JiangMW20}.
For the Hare quota, the group-size requirement already implies $\ell\le k$, so allowing every positive integer $\ell$ leaves standard AJR and SJR unchanged.

We now define how to measure an individual voter’s satisfaction with a committee.
\begin{definition}[Voter’s Satisfaction]
    For election $\mathbb{E}=(N,C,\mathcal{A},k)$ and committee $W\subseteq C$, the \emph{satisfaction} of voter $i\in N$ is
    \[
        S_W(i) = |W \cap A_i|.
    \]
\end{definition}
Likewise, for a group of voters $N'\subseteq N$, we define their total satisfaction as
\[
    S_W(N') = \sum_{i\in N'} S_W(i).
\]

Building on the above concepts, we state two core axioms of justified representation.
\begin{definition}[Strong Justified Representation (SJR)]
    A committee $W\subseteq C$, $|W|=k$, satisfies \emph{Strong Justified Representation} (SJR) if for every integer $1\le \ell\le k$ there exists \emph{no} $\ell$-cohesive group $N'\subseteq N$ such that
    %\[
    %    \exists i\in N',\;S_W(i)<\ell.
    %\]
    $S_W(i)<\ell$ for some $i\in N'$.
    Equivalently, each member of an $\ell$-cohesive group has satisfaction at least $\ell$.
\end{definition}

\begin{definition}[Average Justified Representation (AJR)]
    A committee $W\subseteq C$, $|W|=k$, satisfies \emph{Average Justified Representation} (AJR) if for every integer $1\le \ell\le k$ there exists \emph{no} $\ell$-cohesive group $N'\subseteq N$ such that
    $
        \frac{1}{|N'|}S_W(N')< \ell.
    $
    Equivalently, the average satisfaction of each $\ell$-cohesive group is at least $\ell$.
\end{definition}

Finally, we introduce two decision problems of this work.
\begin{definition}[SJR Committee Existence (\SJRE)]
    Given an election $\mathbb{E}=(N,C,\mathcal{A},k)$, decide whether there exists a committee $W\subseteq C$ with $|W|=k$ that satisfies SJR. The output is “Yes” if such a committee exists, and “No” otherwise.
\end{definition}
\begin{definition}[AJR Committee Existence (\AJRE)]
    Given an election $\mathbb{E}=(N,C,\mathcal{A},k)$, decide whether there exists a committee $W\subseteq C$ with $|W|=k$ that satisfies AJR. The output is “Yes” if such a committee exists, and “No” otherwise.
\end{definition}

%Having established these basic notions, we proceed in Section~\ref{sec:Main Result} to analyze the computational complexity of \SJRE and in Section~\ref{sec:Main Result in AJR} to analyze the computational complexity of \AJRE.
\subsection{Complexity Basics}
We assume the readers are familiar with basic complexity concepts and the complexity classes $\ThetaTwo$ and $\SigmaTwo$.
Readers unfamiliar with these can refer to Sect.~\ref{sec:complexitybasics} where we give a review of these.

Below, we introduce a \(\ThetaTwo\)-complete problem, the Minimum Vertex Cover Inclusion Problem~\cite{hemaspaandra2005complexity}, and a $\SigmaTwo$-complete problem, the \(\neg\forall\exists\) 3-SAT problem, that will be used in our reductions.

\begin{definition}[\vcmember]
    Given an undirected graph \(G=(V,E)\) and a vertex \(x\in V\), decide whether there exists a minimum size vertex cover \(\mathfrak{C}\) of \(G\) such that \(x\in \mathfrak{C}\).
\end{definition}
    
A Boolean formula is in 3-conjunctive normal form (3-CNF) if it is a conjunction of clauses, each consisting of a disjunction of at most three literals.
A literal is a variable or its negation.

\begin{definition}[\(\neg\forall\exists\) 3-SAT Problem (\sat)]
        Let \(\varphi(\mathbf{x},\mathbf{y})\) be a Boolean formula in 3-CNF. 
        The \(\neg\forall\exists\) 3-SAT problem asks whether
        \[
                \exists \mathbf{x}\,\forall \mathbf{y}\; \neg\varphi(\mathbf{x},\mathbf{y}).
        \]
\end{definition}

A 3-CNF formula is \emph{proper} if each clause contains exactly three literals over three distinct variables and is not tautological.
We use the standard restriction to proper 3-CNF instances in the reductions below.

\subsection{On Adding One More Adaptive Query in $\ThetaTwo$}
\label{sec:complexityresult}
In this section, we prove a complexity result used later in the SJR section which may be of independent interest.
Recall that the complexity class $\ThetaTwo$ contains all languages that can be decided in polynomial time with a polynomial number of non-adaptive NP queries (see Sect.~\ref{sec:complexitybasics}).
On the other hand, as we will see in the SJR section, the problem of \SJRE is seemingly ``slightly harder'' than problems in $\ThetaTwo$.
It can be solved by a polynomial number of non-adaptive NP queries \emph{followed by one adaptive NP query that depends on the results of the previous non-adaptive NP queries.}
We will prove in this section that the corresponding problem is still in $\ThetaTwo$, and adding one extra adaptive NP query does not increase any computational power.

We first define the complexity class with the described feature above.

\begin{definition}\label{def:parallel+1}
A language $L$ lies in $\Poly_{\parallel+1}^{\NP}$ if there exist a polynomial $p$ and an $\NP$-oracle Turing
machine $M^{\NP}$ such that, for every input $x$, $M^{\NP}(x)$ runs in time $p(|x|)$ and has the following form:
\begin{enumerate}
        \item $M^{\NP}$ first issues at most $p(|x|)$ \emph{non-adaptive} oracle queries (i.e., the set of these queries
        depends only on $x$), receives their answers, and then
        \item makes one additional oracle query as its \emph{final} step, which may depend on $x$ and on the earlier answers;
        $M^{\NP}$ accepts if and only if this final query is answered \textsf{YES}.
\end{enumerate}
Moreover, $x\in L$ if and only if $M^{\NP}$ accepts $x$.
\end{definition}
In the following, we will prove that $\Poly_{\parallel+1}^\NP=\ThetaTwo$.
At first glance, this may be an easy task: as we mentioned in Sect.~\ref{sec:complexitybasics}, Hemachandra~\cite{hemachandra1987strong} and Buss and Hay~\cite{buss1991truth} independently show that $\Poly_{\parallel}^\NP=\Poly^{\NP[\log n]}$, which indicates that polynomially many non-adaptive NP queries can be simulated by logarithmically many adaptive NP queries; therefore, to show that $\Poly_{\parallel+1}^\NP$ is the same as $\ThetaTwo=\Poly_{\parallel}^\NP=\Poly^{\NP[\log n]}$, we can convert all the non-adaptive queries except for the last one to logarithmically many adaptive queries, and adding one more adaptive query still gives a logarithmic number of queries.
However, the results of Hemachandra~\cite{hemachandra1987strong} and Buss and Hay~\cite{buss1991truth} cannot be directly used here.
The result $\Poly_{\parallel}^\NP=\Poly^{\NP[\log n]}$ only says that, if we can use polynomially many non-adaptive queries to solve a \emph{decision problem}, we can alternatively use logarithmically many adaptive queries to achieve the same task.
In particular, for $\Poly_{\parallel}^\NP$, we get a \emph{binary outcome} after polynomially many parallel queries, and the proofs in Hemachandra~\cite{hemachandra1987strong} and Buss and Hay~\cite{buss1991truth} crucially rely on the fact that this outcome is binary.
On the other hand, in $\Poly_{\parallel+1}^\NP$, supposing $m+1$ queries are made where the first $m$ queries are non-adaptive, there can be $2^m$ different outcomes for the first $m$ queries (in particular, the outcome here is no longer binary), which means there can be up to $2^m$ different possible queries for the last step.
If there is a \emph{binary} ``intermediate outcome'' for the first $m$ non-adaptive queries and the final $(m+1)$-th query only depends on this binary intermediate outcome, then we can apply the results of Hemachandra~\cite{hemachandra1987strong} and Buss and Hay~\cite{buss1991truth}.
However, this ``intermediate outcome'' can be a \emph{string} that is not binary.
In other words, the ``subproblem'' described by the first $m$ queries is not a decision problem, and the results of Hemachandra~\cite{hemachandra1987strong} and Buss and Hay~\cite{buss1991truth} cannot be directly applied.

To prove the following theorem, we use a technique that is different from Hemachandra~\cite{hemachandra1987strong} and Buss and Hay~\cite{buss1991truth}.

\begin{theorem}\label{thm:complexity}
        $\Poly_{\parallel+1}^\NP=\ThetaTwo$.
\end{theorem}
\begin{proof}
        The containment $\ThetaTwo\subseteq \Poly_{\parallel+1}^\NP$ is trivial, and we will prove $\Poly_{\parallel+1}^\NP\subseteq\ThetaTwo$ in the following.

        Fix a language $L\in \Poly_{\parallel+1}^\NP$.
        Consider $M^\NP$ as defined in Definition~\ref{def:parallel+1} and an arbitrary input $x$.
        Let $m+1$ be the number of NP queries made, where the first $m$ queries are non-adaptive.
        Let $\phi_1,\ldots,\phi_m$ be the first $m$ queries, and we assume they are SAT instances without loss of generality.
        Let $h:\{0,1\}^m\to\{0,1\}$ be the function describing the last query.
        In particular, $h(\mathbf{t})=1$ if and only if the last query output ``yes'' when the outcomes of the first $m$ queries are encoded by $\mathbf{t}$ (i.e., $\mathbf{t}$ is defined such that, for each $j=1,\ldots,m$, $\mathbf{t}[j]=1$ if and only if $\phi_j$ is a yes instance).
        Since the last query is an NP query, the evaluation of $h(\cdot)$ is in NP, i.e., for a given $\mathbf{t}$, deciding if $h(\mathbf{t})=1$ is in NP.
        Notice that $m,\phi_1,\ldots,\phi_m$, and $h(\cdot)$ depend on $x$ and $m<p(|x|)$ ($p$ is the polynomial time complexity upper bound as defined in Definition~\ref{def:parallel+1}).

        We consider the following two decision problems:
        \begin{enumerate}
                \item $A(x,i)$: decide if at least $i$ of the $m$ SAT instances $\phi_1,\ldots,\phi_m$ are yes instances;
                \item $B(x,i)$: decide if there exists $\mathbf{t}\in\{0,1\}^m$ with Hamming weight (the number of entries equal to $1$) at least $i$ such that i) $h(\mathbf{t})=1$ and ii) for each $j=1,\ldots,m$, $\mathbf{t}[j]=1$ implies that $\phi_j$ is a yes instance.
        \end{enumerate}
        Notice that, for the second problem, $\mathbf{t}[j]=1$ must imply that $\phi_j$ is a yes instance, while $\phi_j$ can be a yes or a no instance for $\mathbf{t}[j]=0$.
        It is easy to see that both decision problems above (with $i$ given) are in NP.

        Both predicates remain true when $i$ decreases.
        Since $A(x,0)$ is a yes instance, there are unique thresholds $i_A\in\{0,\ldots,m\}$ and $i_B\in\{-1,0,\ldots,m\}$ such that, for every $i\in\{0,\ldots,m\}$,
        \[
            A(x,i)=\textsf{YES}\iff i\le i_A,
            \qquad
            B(x,i)=\textsf{YES}\iff i\le i_B.
        \]
        For $B$, first query $B(x,0)$.
        If the answer is no, set $i_B=-1$.
        Moreover, we have $i_A\geq i_B$: notice that if we have a vector $\mathbf{t}$ with Hamming weight $i$ that satisfies the conditions of problem $B$, it implies there are at least $i$ yes instances among $\phi_1,\ldots,\phi_m$.

        Next, we show that $x\in L$ if and only if $i_A=i_B$.
        %Let $S\subseteq\{1,\ldots,m\}$ be the set of indices corresponding to yes instances in $\phi_1,\ldots,\phi_m$.
     % We have $i_A=|S|$.
        Given a vector $\mathbf{t}\in\{0,1\}^m$, let $S_{\mathbf{t}}\subseteq\{1,\ldots,m\}$ be the set of entry indices $j$ of $\mathbf{t}$ with $\mathbf{t}[j]=1$.
        If $\mathbf{t}$ satisfies the two conditions in problem $B$, we must have $|S_{\mathbf{t}}|\leq i_B$.
        Let $\mathbf{t}^\ast$ be the actual vector that encodes the satisfiability of the SAT instances $\phi_1,\ldots,\phi_m$.
        We have $i_A=|S_{\mathbf{t}^\ast}|$, and we additionally have $S_{\mathbf{t}}\subseteq S_{\mathbf{t}^\ast}$ for every $\mathbf{t}$ satisfying the two conditions in problem $B$.
        If $i_A>i_B$, we know $h(\mathbf{t}^\ast)=0$, since $h(\mathbf{t}^\ast)=1$ requires $|S_{\mathbf{t}^\ast}|\leq i_B$ and this contradicts $i_B<i_A=|S_{\mathbf{t}^\ast}|$.
        This implies $x\notin L$.
        If $i_A=i_B$, there exists $\mathbf{t}^\dag$ that satisfies the two conditions in problem $B$, and moreover $|S_{\mathbf{t}^\dag}|=i_B$.
        Since $S_{\mathbf{t}^\dag}\subseteq S_{\mathbf{t}^\ast}$ and $i_A=i_B$, it can only be the case that $\mathbf{t}^\dag=\mathbf{t}^\ast$.
        By our definition, we have $h(\mathbf{t}^\dag)=1$, so $h(\mathbf{t}^\ast)=1$, implying $x\in L$.

        Finally, to decide if $x\in L$, we only need to compute the value of $i_A$ and $i_B$.
        If $m=0$, then $i_A=0$ and the single query $B(x,0)$ determines whether $i_B=0$ or $i_B=-1$.
        If $m>0$, the values can be computed by two binary searches after the initial query to $B(x,0)$.
        Therefore, we have used the NP oracle for at most $2\lceil\log(m+1)\rceil+1\leq O(\log(p(|x|)))$ times. 
        By Definition~\ref{def:thetatwo}, $L\in\ThetaTwo$.
\end{proof}

%As the first remark, the complexity class is still $\ThetaTwo$ if we allow a constant number of adaptive queries at the end (instead of just one).
%The same argument above can apply.
%As the second remark, 
As a remark, our proof can be adapted to show the result $\Poly_{\parallel}^\NP=\Poly^{\NP[\log n]}$.
The only part that needs to be adapted is the function $h$.
In our case, $h$ represents the last NP query, and it needs to be evaluated by an NP oracle.
In the case proving $\Poly_{\parallel}^\NP=\Poly^{\NP[\log n]}$, $h$ represents the execution of the Turing machine after polynomially many non-adaptive queries, and $h$ can be computed in polynomial time.
Other than this, the remaining part of the proof is unchanged.
In fact, our argument is much simpler than the ideas of ``mind-changing counts'' used in Hemachandra~\cite{hemachandra1987strong} and Buss and Hay~\cite{buss1991truth}.

%We write,
%    \[
%        \Theta_2^p \;=\; \Poly^{\NP[\log n]}
%        \;=\; \Poly^{\NP}_{\parallel},
%    \]
%where $\Poly^{\NP[\log n]}$ denotes deterministic polynomial-time machines with $O(\log n)$ adaptive queries to an NP oracle, and $\Poly^{\NP}_{\parallel}$ denotes those with a single parallel round of polynomially many NP-oracle queries.

\section{Complexity Result for AJR}
\label{sec:Main Result in AJR}
In this section, we prove the following result.

\begin{theorem}\label{thm:AJR}
        \AJRE is \(\SigmaTwo\)-complete.
\end{theorem}

\begin{lemma}\label{lem:AJRmem}
    \AJRE is in \(\SigmaTwo\).
\end{lemma}
\begin{proof}
Membership in $\SigmaTwo$ is straightforward.
The existential certificate ($y$ in Definition~\ref{def:sigmatwo}) encodes an AJR committee $W$.
The universal certificate ($z$ in Definition~\ref{def:sigmatwo}) encodes an arbitrary pair $(\ell,N')$.
If $\ell\notin\{1,\ldots,k\}$, the verifier accepts automatically.
Otherwise, it first checks whether $N'$ is an $\ell$-cohesive group; if not, it accepts automatically.
Otherwise, it checks whether the average satisfaction of $N'$ under $W$ is at least $\ell$.
\end{proof}
In the remaining part of this section, we prove the $\SigmaTwo$-hardness of \AJRE.

In Sect.~\ref{sec:setsystem}, we construct a set system that will be used as a gadget in our reduction.
%We introduce two intermediate problems \lamAJRExist and \lamAJRBExist and present the chain of reductions in Sect.~\ref{sec:AJR_sattobase} and Sect.~\ref{sec:AJR_basetoAJR}.
%The latter SJR proof follows the same base-removal and quota-adjustment template, replacing average satisfaction by individual satisfaction.
Sect.~\ref{sec:AJR_sattobase} and Sect.~\ref{sec:AJR_basetoAJR} aim to prove the $\SigmaTwo$-hardness of \AJRE.
To prove the $\SigmaTwo$-hardness, we introduce two intermediate problems \lamAJRExist and \lamAJRBExist, and we prove the following chain of Karp reductions:
\[
\text{\sat}\leq_K\text{\lamAJRBExist}\leq_K\text{\lamAJRExist}\leq_K\text{\AJRE}.
\]
We will prove $\text{\sat}\leq_K\text{\lamAJRBExist}$ in Sect.~\ref{sec:AJR_sattobase} and $\text{\lamAJRBExist}\leq_K\text{\lamAJRExist}\leq_K\text{\AJRE}$ in Sect.~\ref{sec:AJR_basetoAJR}.

We will formally define \lamAJRExist and \lamAJRBExist soon.
Before this, we describe some high-level ideas on these two decision problems and the intuitive reason why we introduce them.
\lamAJRExist is a variant of \AJRE where $\ell$-cohesive groups in the definition of AJR are replaced by $(\lambda,\ell)$-groups (Definition~\ref{def:lambdaellgroup}).
In other words, instead of setting the size of the $1$-cohesive group to $n/k$, we allow a customized size of $\lambda$ where $\lambda$ is specified as an input.
\lamAJRBExist is a variant where, in addition to the customized cohesive group sizes, each voter $i$ also has a customized ``base satisfaction'' $s(i)$ such that $i$ has a base satisfaction of $s(i)$ when she approves no candidate in the committee and $i$'s satisfaction is redefined to $S_W(i)+s(i)$ for committee $W$.
The first part of the reduction chain, $\text{\sat}\leq_K\text{\lamAJRBExist}$ contains the essential ideas on the hardness of our problem.
For the second part of the reduction chain $\text{\lamAJRBExist}\leq_K\text{\lamAJRExist}\leq_K\text{\AJRE}$, it is achieved by carefully adding dummy voters and candidates.
We use these two intermediate problems to increase the accessibility of our proof.

\begin{definition}[$\lambda$-Average Justified Representation (\lamAJR)]\label{def:lamAJR}
    A committee $W\subseteq C$ with $|W|=k$ satisfies $\lambda$-\emph{Average Justified Representation} if for every positive integer $\ell$ there exists \emph{no} $(\lambda,\ell)$-group $N'\subseteq N$ such that
    $
        \frac{1}{|N'|}S_W(N')< \ell.
    $
    Equivalently, the average satisfaction of each \((\lambda,\ell)\)-group is at least \(\ell\).
\end{definition}

\begin{definition}[\lamAJR Committee with base satisfactions (\lamAJRB)]
    Given an election $\mathbb{E}=(N,C,\mathcal{A},k)$, and a base satisfaction function $s:N\to\mathbb{Z}_{\geq0}$, 
    a committee $W\subseteq C$ with $|W|=k$ satisfies \(\lambda\)-\emph{Average Justified Representation with base satisfactions} (\lamAJRB) 
    if for every positive integer $\ell$ there exists \emph{no} $(\lambda,\ell)$-group $N'\subseteq N$ such that
    \[
        \frac{1}{|N'|}\left(S_W(N')+\sum_{i\in N'}s(i)\right)< \ell.
    \]
    Equivalently, the average satisfaction of each \((\lambda,\ell)\)-group is at least \(\ell\), when each voter $i$ has a base satisfaction $s(i)$.
\end{definition}

%\paragraph{Auxiliary input restrictions.}
The auxiliary \(\lambda\)-AJR problems used in this paper require the input parameter \(\lambda\) to satisfy \(2<\lambda\le |N|\).
For the base-satisfaction variant, we additionally require
\[
    \sum_{i\in N}s(i)\le |\mathbb{E}|^{O(1)},
\]
and hence every value \(s(i)\) is polynomially bounded.
These input restrictions are part of the definitions of \lamAJRExist and \lamAJRBExist below.

\begin{definition}[\lamAJR Committee Existence (\lamAJRExist)]\label{def:lamAJRE}
    Given an election $\mathbb{E}=(N,C,\mathcal{A},k)$ and an integer parameter \(2<\lambda\le |N|\), decide whether there exists a committee $W\subseteq C$ with $|W|=k$ that satisfies \lamAJR.
    The output is “Yes” if such a committee exists, and “No” otherwise.
\end{definition}

Next, we define a variant of \lamAJRExist where each voter $i$ has a ``base satisfaction'' $s(i)$ (given as an input parameter) such that her satisfaction for committee $W$ is redefined to $S_W(i)+s(i)$.

\begin{definition}[\lamAJRB Committee Existence (\lamAJRBExist)]
    Given an election $\mathbb{E}=(N,C,\mathcal{A},k)$, an integer parameter \(2<\lambda\le |N|\), and a base-satisfaction function $s:N\to\mathbb{Z}_{\geq0}$ with polynomially bounded total value, 
    decide whether there exists a committee $W\subseteq C$ with $|W|=k$ that satisfies \lamAJRB.
    The output is ``Yes'' if such a committee exists, and ``No'' otherwise.
\end{definition}

\subsection{A Set System}
\label{sec:setsystem}

%Let \(\mathrm{Var}(\varphi)\) denote the set of variables occurring in a formula \(\varphi\).

A \emph{set system} (or \emph{family of sets}) is a pair \((X,\mathcal{F})\), where 
\(X\) is a non-empty ground set and \(\mathcal{F}\subseteq 2^X\) is a collection of subsets of \(X\).

In Lemma~\ref{lem:setsystem}, we present the construction of a set system $(X,\mathcal{F})$ such that $\mathcal{F}$ consists of $2n$ sets and can be partitioned into $n$ groups of two sets.
Suppose we want to choose $n$ out of $2n$ sets in $\mathcal{F}$ with the maximum size intersection.
Lemma~\ref{lem:setsystem} says that the choice is optimal if and only if we select exactly one set from each of the $n$ groups.

Notably, our construction only requires $O(n^2)$ elements in $X$, which makes it useful for reductions.
A na\"{\i}ve construction can be the following.
Let $X$ be the set of all binary strings of length $n$.
For the $i$-th group of two sets in $\mathcal{F}$, let one set contain all the strings except those with the $i$-th bit $0$, and let the other set contain all the strings except those with the $i$-th bit $1$.
It is easy to see that the intersection is maximized if and only if we choose one set from each group: considering the complement of each of the $2n$ sets, the union is minimized when we choose one set from each group (if not choosing in this way, the union becomes the entire $X$).
However, this construction requires an exponential number ($2^n$ to be precise) of elements.

Intuitively, the $n$ groups will become $n$ variables in a 3-CNF formula, so choosing $n$ sets such that exactly one set is chosen in each group corresponds to a truth assignment of the $n$ variables.
We will suitably set the size of the cohesive group such that only the choice of $n$ sets with the maximum intersection can be large enough to form a cohesive group.

\begin{lemma}\label{lem:setsystem}
    For any integer $n\ge 2$, there exists a set system $(X,\mathcal{F})$ where $|\mathcal{F}|=2n$, $|X|=4n(n-1)$, and $(\mathcal{F}_1,\ldots,\mathcal{F}_n)$ is a partition of $\mathcal{F}$ with $|\mathcal{F}_1|=\cdots=|\mathcal{F}_n|=2$, and this partition system satisfies the following properties:
    \begin{itemize}
        \item for any $\mathcal{S}\subseteq\mathcal{F}$ satisfying $|\mathcal{S}\cap\mathcal{F}_i|=1$ for each $i=1,\ldots,n$, we have $|\bigcap_{F\in\mathcal{S}}F|=n(n-1)$, and
        \item for any $\mathcal{S}\subseteq\mathcal{F}$ with $|\mathcal{S}|=n$ such that $|\mathcal{S}\cap\mathcal{F}_i|\neq1$ for some $i$, we have $|\bigcap_{F\in\mathcal{S}}F|<n(n-1)$.
    \end{itemize}
\end{lemma}
\begin{proof}
    \noindent\textbf{Construction.}
    For each ordered pair \((i,j)\) with \(i\neq j\) in \([n]\), create a 4-element block
    \(C^{ij}=\{c^{ij}_1,c^{ij}_2,c^{ij}_3,c^{ij}_4\}\). Define
    \[
    X=\bigcup_{\substack{i\neq j\\ i,j\le n}} C^{ij},\qquad
    \mathcal{F}=\{F_1,\bar F_1,\dots,F_n,\bar F_n\},\qquad
    \mathcal{F}_k=\{F_k,\bar F_k\}.
    \]

    For each \(k\in[n]\),
    \[
    \begin{aligned}
    F_k &=
        \bigcup_{\substack{i\neq j\\ k\notin\{i,j\}}} C^{ij}
        \;\cup\; \bigcup_{j\neq k}\{c^{kj}_1,c^{kj}_2\}
        \;\cup\; \bigcup_{i\neq k}\{c^{ik}_1,c^{ik}_3\}, \\[2mm]
    \bar F_k &=
        \bigcup_{\substack{i\neq j\\ k\notin\{i,j\}}} C^{ij}
        \;\cup\; \bigcup_{j\neq k}\{c^{kj}_3,c^{kj}_4\}
        \;\cup\; \bigcup_{i\neq k}\{c^{ik}_2,c^{ik}_4\}.
    \end{aligned}
    \]

    Since the \(C^{ij}\) are disjoint, \(|X|=4n(n-1)\) and \(|\mathcal{F}|=2n\) as required.

    Each block $C^{ij}$ encodes the four possible choices of one set from $\mathcal F_i$ and one set from $\mathcal F_j$.
    Choosing exactly one set from every pair leaves exactly one element in each block.
    A noncanonical selection eliminates blocks involving a pair from which both sets are chosen; the count below shows that this loss exceeds the gain from omitted pairs.

    \noindent\textbf{Preliminary Observation.}
    As \(X=\bigcup_{\substack{i\neq j\\ i,j\le n}} C^{ij}\),
    \begin{align*}
    \bigcap_{F\in\mathcal{S}}F
        &= \left(\bigcup_{\substack{i\neq j\\ i,j\le n}} C^{ij}\right)\cap \bigcap_{F\in\mathcal{S}}F= \bigcup_{\substack{i\neq j\\ i,j\le n}} \bigcap_{F\in\mathcal{S}}(F\cap C^{ij})    \\
        %&= \bigcup_{\substack{i\neq j\\ i,j\le n}} \bigcap_{F\in\mathcal{S}}(F\cap C^{ij}) \\
        &= \bigcup_{i\neq j}
            \left(
            C^{ij}\ \cap
            \bigcap_{\substack{F_m\in\mathcal S\\ F_m\in\mathcal F_i\cup \mathcal F_j}} F_m
            \right). \tag{$\triangle$}
    \end{align*}
    For convenience, set
    $\displaystyle
    S_{ij}:=C^{ij}\ \cap
            \bigcap_{\substack{F_m\in\mathcal S\\ F_m\in\mathcal F_i\cup \mathcal F_j}} F_m.
    $
    
    Thus, it suffices to characterize \(F_k\cap C^{ij}\) and \(\bar F_k\cap C^{ij}\). Fix a block \(C^{ij}\). For any \(k\),
    \[
    \begin{array}{c|c}
    k\notin\{i,j\} & F_k\cap C^{ij}=\bar F_k\cap C^{ij}=C^{ij};\\
    k=i                        & F_i\cap C^{ij}=\{c^{ij}_1,c^{ij}_2\},\;\; \bar F_i\cap C^{ij}=\{c^{ij}_3,c^{ij}_4\};\\
    k=j                        & F_j\cap C^{ij}=\{c^{ij}_1,c^{ij}_3\},\;\; \bar F_j\cap C^{ij}=\{c^{ij}_2,c^{ij}_4\}.
    \end{array}
    \tag{$\star$}
    \]

    \noindent\textbf{Property 1.}
    Let \(\mathcal{S}\subseteq\mathcal{F}\) contain exactly one set from each \(\mathcal{F}_k\).
    By ($\triangle$) and ($\star$), each \(S_{ij}\) is a singleton
    \(\{c^{ij}_1\},\{c^{ij}_2\},\{c^{ij}_3\},\) or \(\{c^{ij}_4\}\),
    depending on whether \(F_i\) or \(\bar F_i\), and \(F_j\) or \(\bar F_j\), is chosen.
    Thus, each block contributes exactly one element, giving
    \(\bigl|\bigcap_{F\in\mathcal{S}}F\bigr|=n(n-1)\).

    \noindent\textbf{Property 2.}
    Now let \(\mathcal{S}\subseteq\mathcal{F}\) with \(|\mathcal{S}|=n\), and assume
    \(|\mathcal{S}\cap \mathcal{F}_k|\neq 1\) for some \(k\).
    For each \(r\), set \(a_r:=|\mathcal{S}\cap \mathcal{F}_r|\in\{0,1,2\}\), and define
    \[
    \alpha=|\{r:a_r=0\}|,\quad
    \beta=|\{r:a_r=1\}|,\quad
    \gamma=|\{r:a_r=2\}|.
    \]
    Since \(\sum_r a_r=n\) and \(\beta+2\gamma=n\), we have \(\alpha=\gamma\).

    For each pair \((i,j)\), ($\triangle$) and ($\star$) imply
    \[
    |S_{ij}| =
    \begin{cases}
    0 & a_i=2 \text{ or } a_j=2,\\
    1 & a_i=a_j=1,\\
    2 & \{a_i,a_j\}=\{0,1\},\\
    4 & a_i=a_j=0.
    \end{cases}
    \]
    Summing over ordered pairs gives
    \[
    \Bigl|\bigcap_{F\in\mathcal{S}}F\Bigr|
        = \beta(\beta-1) + 4\beta\alpha + 4\alpha(\alpha-1).
    \]
    Substituting \(\beta=n-2\gamma\) and \(\alpha=\gamma\) yields
    \[
    \Bigl|\bigcap_{F\in\mathcal{S}}F\Bigr| = n(n-1)-2\gamma.
    \]
    Since some \(a_k\neq 1\), we have \(\gamma\ge 1\), and therefore
    \(|\bigcap_{F\in\mathcal{S}}F|\le n(n-1)-2<n(n-1)\).
\end{proof}

\subsection{Reduction from \sat to \lamAJRBExist}
\label{sec:AJR_sattobase}

We prove the following lemma in this section.
The high-level idea is to separate two roles that are easy to conflate.
For each variable, we create assignment candidates, which are used only to certify that a large cohesive group corresponds to a satisfying full assignment.
For each existential variable, we additionally create blocker candidates, and these are the only non-dummy candidates that a feasible committee is forced to select.
The separation lets the soundness proof reason about the individual satisfaction of the canonical group, rather than only about its common approval set.

%\biaoshuai{I do not get the descriptions above. It would be great if we could write a few paragraphs to describe the ideas in more detail, as we did for the SJR part.}

\begin{lemma}\label{lem:AJR}
    \sat is Karp reducible to \lamAJRBExist.
\end{lemma}

We prove this lemma in the remaining part of this section.

\paragraph{The construction.}
Recall that an instance of the \sat can be assumed to consist of a proper 3-CNF formula $\varphi$ where variables are partitioned into $\mathbf{x}=(x_1,\ldots,x_{n_1})$ and $\mathbf{y}=(y_1,\ldots,y_{n_2})$.
It is a yes instance if there exists a truth assignment of variables in $\mathbf{x}$ such that $\varphi$ is \emph{unsatisfied} for every truth assignment of variables in $\mathbf{y}$.
It is a no instance if for all truth assignments of variables in $\mathbf{x}$ there exists a truth assignment of variables in $\mathbf{y}$ such that $\varphi$ is satisfied.
Let $n=n_1+n_2$ and $m$ be the number of clauses.
We assume without loss of generality that $n>2$ and every existential variable appears in some clause; unused existential variables can be removed.
Let $z_i=x_i$ for $i\le n_1$ and $z_{n_1+j}=y_j$ for $j\le n_2$.
For a truth value $t\in\{T,F\}$, let $\overline{t}$ denote the opposite truth value.

Given a \sat instance $\varphi$, we construct an instance of \lamAJRBExist as follows.
Set
\[
    \chi=n,\qquad \gamma=7m\chi+1,\qquad
    \lambda=\frac{m\chi}{n}+(n-1)\gamma+1=m+(n-1)\gamma+1.
\]

We first define the candidate set.
For every variable $z_i$, introduce assignment candidates
\[
    A_i=\{a_i^T,a_i^F\}.
\]
These candidates are used by the verifier voters only.
For every existential variable $x_i$, $i\le n_1$, introduce blocker candidates
\[
    B_i=\{b_i^T,b_i^F\}.
\]
Finally, introduce dummy candidates
\[
    \varepsilon=\{\varepsilon_1,\ldots,\varepsilon_n\}.
\]
Let $A=\bigcup_{i=1}^n A_i$ and $B=\bigcup_{i=1}^{n_1}B_i$, and set the candidate set to $C=A\cup B\cup\varepsilon$.
The committee size is
\[
    k=n_1+n.
\]

The voters consist of five sets: $U$, $P$, $Q$, $R$, and $D$.
For each clause $j$ and each satisfying partial assignment $\sigma$ to the three variables of that clause, create $\chi$ identical voters.
The set $U_j$ contains all voters created for clause $j$, and $U=\bigcup_{j=1}^m U_j$.
A voter in $U_j$ associated with $\sigma$ approves:
\begin{itemize}
    \item both assignment candidates in $A_i$ for every variable $z_i$ not appearing in clause $j$;
    \item the assignment candidate $a_i^{\sigma_i}$ for every variable $z_i$ appearing in clause $j$;
    \item the opposite blocker $b_i^{\overline{\sigma_i}}$ for every existential variable $x_i$ appearing in clause $j$.
\end{itemize}
Since each clause has $7$ satisfying partial assignments (there are $8$ partial assignments in total; the assignment making all three literals false is the only unsatisfying assignment), we have $|U_j|=7\chi$ and $|U|=7m\chi$.

For each $i\le n_1$, let $P_i$ consist of $\lambda$ voters approving exactly $B_i=\{b_i^T,b_i^F\}$, and set $P=\bigcup_{i=1}^{n_1}P_i$.
For each $j\in[n]$, let $D_j$ consist of $\lambda$ voters approving exactly $\varepsilon_j$, and set $D=\bigcup_{j=1}^{n}D_j$.
Let $Q=\{q_1,\ldots,q_n\}$, where every voter in $Q$ approves all assignment candidates in $A$ and no blocker or dummy candidate.

Finally, construct $R$ based on the set system in Lemma~\ref{lem:setsystem}.
Given the set system $(X,\mathcal{F})$, identify each pair $\mathcal{F}_i=\{F_i,\bar F_i\}$ with the assignment-candidate pair $A_i$.
For each element $z\in X$, create $\gamma$ identical voters.
Such a voter approves $a_i^T$ if $z\in F_i$, and approves $a_i^F$ if $z\in\bar F_i$.
An $R$-voter may approve both assignment candidates in a pair $A_i$.
Since $R$-voters approve no blocker candidates, they receive no satisfaction from the candidates encoding the committee's existential assignment.
For a noncanonical choice of $n$ assignment candidates, Lemma~\ref{lem:setsystem} bounds the number of $R$-voters approving them all; the proof below shows that these voters, together with $U\cup Q$, cannot reach the required group size.

This completes the description of the voters.
All blocks have polynomial size: $\chi$, $\gamma$, $\lambda$, and the base satisfactions below have polynomially bounded total value in $n$ and $m$, and the voters are explicitly expanded.
Moreover, \(\lambda\le |N|\), since \(R\cup Q\) already contains at least \(\lambda\) voters for \(n>2\).
The base satisfaction function is
\begin{equation}\label{eqn:AJR-base-satisfaction}
    s(v)=\begin{cases}
        n-1 & \text{if } v\in U,\\
        n & \text{if } v\in R \text{ or } v\in Q \setminus \{q_n\},\\
        n+m\chi-\chi & \text{if } v=q_n,\\
        0 & \text{if } v\in P\cup D.
    \end{cases}
\end{equation}

\begin{figure}[ht!]
\centering
\resizebox{0.98\textwidth}{!}{
\begin{tikzpicture}[
        thick,
        >={Stealth[length=2mm]},
        vertex/.style={circle, draw=black, fill=white, inner sep=0pt, minimum size=0.55cm, font=\small\bfseries},
        voter/.style={rectangle, draw=black!65, fill=gray!5, rounded corners=2pt, minimum width=1.5cm, minimum height=0.6cm, font=\scriptsize, align=center},
        candidate/.style={circle, draw=black!80, fill=white, minimum size=0.62cm, font=\scriptsize, align=center},
        candset/.style={rectangle, draw=black!80, fill=white, rounded corners=2pt, minimum width=1.4cm, minimum height=0.55cm, font=\scriptsize, align=center},
        edge_app/.style={draw=blue!70, thick, ->},
        target_app/.style={draw=red!80, thick, ->},
        budget_app/.style={draw=black!65, thick, dashed, ->},
        forced_app/.style={draw=gray!75, thick, ->},
        cert_a/.style={draw=green!60!black, thick, ->},
        cert_ell/.style={draw=orange!90!black, thick, ->}
]

% ==========================================
% 1. Input SAT Instance
% ==========================================
\node[font=\bfseries] at (-2.5, 7.5) {\sat Instance $\varphi$};
\node[align=center, font=\small, draw=black!50, rounded corners, fill=yellow!5, inner sep=6pt] at (-2.5, 5.2) {
    $\exists \mathbf{x}\,\forall \mathbf{y}\;\neg\varphi(\mathbf{x},\mathbf{y})$\\
    \vspace{0.1cm}\\
    Clause $C_j$ with satisfying\\ partial assignment $\sigma$ and let $x_\ell$ \\
    be an existential variable in $C_j$\\
};

% ==========================================
% 2. Divider & Structure Layout
% ==========================================
\draw[gray!30, line width=1.5pt] (0.5, 7.8) -- (0.5, -3.5);
\node[font=\bfseries, anchor=west] at (0.8, 7.5) {One election instance for $\varphi$};

% Section Labels (Anchored to the left)
\node[font=\scriptsize\bfseries, black!70, align=left, anchor=west] at (0.8, 6.8) {Blocker Forcing\\(Existential only)};
\node[font=\scriptsize\bfseries, blue!80, align=left, anchor=west] at (0.8, 4.5) {Clause Verifier\\(For clause $j$, asst $\sigma$)};
\node[font=\scriptsize\bfseries, green!50!black, align=left, anchor=west] at (0.8, 0.7) {Assignment\\Regulators};
\node[font=\scriptsize\bfseries, gray!90, align=left, anchor=west] at (0.8, -2.0) {Forced Dummies};

% ==========================================
% 3. Voters (Shifted right to x = 7.0)
% ==========================================
\node[voter, label=left:forced] (vP) at (7.0, 6.8) {$P_\ell$ \\ $|P_\ell|=\lambda$};
\node[voter, draw=blue!70, line width=1pt, label=left:\textcolor{blue!80}{verifier}] (vU) at (7.0, 4.5) {$U_j^\sigma$ \\ $|U_j^\sigma|=\chi$};

\node[voter, label=left:boost] (vQ) at (7.0, 1.5) {$Q$ \\ $|Q|=n$};
\node[voter, label=left:\textcolor{orange!90!black}{set system}] (vR) at (7.0, 0.0) {$R$ \\ $\gamma$ each};

\node[voter, label=left:forced] (vD) at (7.0, -2.0) {$D_j$ \\ $|D_j|=\lambda$};

% ==========================================
% 4. Candidates (Shifted to x = 11.0)
% ==========================================
% B_i Candidates & Box
\node[candidate, fill=red!10] (bT) at (11.0, 6.8) {$b_\ell^T$};
\node[candidate, fill=red!10] (bF) at (12.5, 6.8) {$b_\ell^F$};
\node[draw=red!80, dashed, rounded corners, fit=(bT) (bF), inner sep=5pt] (bbox) {};
\node[font=\scriptsize\bfseries, red!80, anchor=west] at ([xshift=4pt]bbox.east) {$B_\ell$ (for $x_\ell$)};

% Giant B box background spanning all blockers (Pulled up safely)
\coordinate (B_safe_bound) at (15.7, 5.8);
\begin{scope}[on background layer]
    \node[draw=red!70, dashed, rounded corners, fit=(bbox) (B_safe_bound), inner sep=8pt, fill=red!2] (Bsuper) {};
    \node[font=\scriptsize\bfseries, red!70!black, anchor=south east, xshift=-15pt, yshift=2pt] at (Bsuper.south east) {$B=B_1\cup B_2\cup \cdots \cup B_{n_1}$};
\end{scope}

% A_i Candidates & Box (in clause) - Moved DOWN to y=4.2
\node[candidate, fill=blue!10] (aT) at (11.0, 4.2) {$a_i^T$};
\node[candidate, fill=blue!10] (aF) at (12.5, 4.2) {$a_i^F$};
\node[draw=blue!80, dashed, rounded corners, fit=(aT) (aF), inner sep=5pt] (abox_in) {};
\node[font=\scriptsize\bfseries, blue!80, anchor=west] at ([xshift=4pt]abox_in.east) {$A_i$ ($z_i \in C_j$)};

% A_k Candidates & Box (not in clause) - Moved DOWN to y=2.7
\node[candidate, fill=blue!10] (akT) at (11.0, 2.7) {$a_k^T$};
\node[candidate, fill=blue!10] (akF) at (12.5, 2.7) {$a_k^F$};
\node[draw=blue!80, dashed, rounded corners, fit=(akT) (akF), inner sep=5pt] (abox_out) {};
\node[font=\scriptsize\bfseries, blue!80, anchor=west] at ([xshift=4pt]abox_out.east) {$A_k$ ($z_k \notin C_j$)};

% Giant A box background spanning all assignments (Shifted down)
\coordinate (A_safe_bound) at (15.7, 1.2);
\begin{scope}[on background layer]
    \node[draw=green!50!black, dashed, rounded corners, fit=(abox_in) (abox_out) (A_safe_bound), inner sep=10pt, fill=green!2] (Asuper) {};
    \node[font=\scriptsize\bfseries, green!50!black, anchor=south east, xshift=-15pt, yshift=4pt] at (Asuper.south east) {$A=A_1\cup A_2\cup \cdots \cup A_n$};
\end{scope}

% Dummy Candidates
\node[candset, fill=gray!12] (eps) at (11.75, -2.0) {$\varepsilon_j$};
\node[font=\scriptsize\bfseries, gray!90, anchor=west] at ([xshift=4pt]eps.east) {$\varepsilon$ (Dummies)};

% ==========================================
% 5. Routed Approval Arrows (Collision-Free)
% ==========================================
% --- Blocker Forcing (P_i) ---
\draw[budget_app] (vP.east) -- (bbox.west);

% --- Clause Verifier (U_j^\sigma) ---
\draw[target_app] (vU.east) to[out=25, in=200] node[pos=0.45, fill=white, inner sep=1pt, font=\tiny, text=red!80] {opposite $b_\ell^{\overline{\sigma_\ell}}$} (bF.south west);
\draw[edge_app] (vU.east) to[out=-5, in=180] node[pos=0.5, fill=white, inner sep=1pt, font=\tiny, text=blue!80] {matching $a_i^{\sigma_i}$} (aT.west);
\draw[edge_app] (vU.east) to[out=-25, in=180] node[pos=0.55, fill=white, inner sep=1pt, font=\tiny, text=blue!80] {both candidates in $A_k$} (abox_out.west);

% --- Assignment Regulators (Q, R) ---
\draw[cert_a] (vQ.east) to[out=10, in=200] node[pos=0.7, fill=white, inner sep=1pt, font=\tiny, text=green!50!black] {approves all $A$} (Asuper.west);
\draw[cert_ell] (vR.east) to[out=15, in=250] node[pos=0.7, fill=white, inner sep=1pt, font=\tiny, text=orange!90!black] {$\subseteq A$} (Asuper.south);

% --- Forced Dummies (D_j) ---
\draw[forced_app] (vD.east) -- (eps.west);

% % Footer Notes
% \node[anchor=west, font=\footnotesize] at (0.8, -4.2) {This figure illustrates the connections for a single verifier block $U_j^\sigma$.};
% \node[anchor=west, font=\scriptsize, text=gray] at (0.8, -4.6) {Note: $U_j^\sigma$ links satisfying clause assignments to candidates and opposite blockers.};

\end{tikzpicture}
}
\caption{The reduction from \sat to \lamAJRBExist (Lemma~\ref{lem:AJR}).}
\label{fig:AJRconstruction}
\end{figure}

\paragraph{High-level intuitions.}
We will describe the main idea of the reduction here.
Starting from an instance
\(\varphi(x,y)\) of \sat, we want committees in the
constructed election to encode assignments to the existential variables
\(x=(x_1,\ldots,x_{n_1})\).  For each existential variable \(x_i\), we create a
pair of blocker candidates
$B_i=\{b_i^T,b_i^F\}$.
The voters in \(P_i\) approve exactly these two blocker candidates, so every
feasible \lamAJRB committee must select at least one candidate from
\(B_i\).  We also create dummy candidates \(\epsilon_1,\ldots,\epsilon_n\),
each forced by a corresponding voter block \(D_j\).  Since the committee size is
set to
$k=n_1+n$,
after selecting all forced dummy candidates there is room for exactly one
blocker candidate from each \(B_i\).  Thus, any feasible committee naturally
encodes an assignment \(x^*\) to the existential variables, where
\(W\cap B_i=\{b_i^{x_i^*}\}\).

The main difficulty is to make this encoded assignment interact with the
universal variables \(y\).  For this purpose, we separate the candidates used to
encode committees from the candidates used to certify large cohesive groups.
For every variable \(z_i\), including both existential and universal variables,
we introduce an assignment-candidate pair
$A_i=\{a_i^T,a_i^F\}$.
These assignment candidates are not intended to be selected by the committee.
Instead, they serve as labels for potential large cohesive groups.  The voters
\(Q\) and \(R\) are designed so that any sufficiently large group with \(n\)
common assignment candidates must choose exactly one candidate from each pair
\(A_i\), and therefore corresponds to a full truth assignment to all variables.
The set-system gadget from Section~3.1 is used precisely to enforce this
canonical structure: if a group tries to use a non-canonical choice, for example
choosing both candidates from one pair and none from another, then too few
\(R\)-voters can participate, and the group cannot reach the required size
threshold.

The clause voters \(U\) then connect these canonical groups to the formula
\(\varphi\).  For each clause and each satisfying partial assignment
\(\sigma\) to the three variables in that clause, we create \(\chi\) identical
voters.  Such a voter approves the assignment candidates compatible with
\(\sigma\), approves both assignment candidates for variables outside the
clause, and, for every existential variable \(x_i\) appearing in the clause,
approves the opposite blocker \(b_i^{\overline{\sigma_i}}\).  Consequently, a
canonical \((\lambda,n)\)-group exists exactly for a satisfying full assignment:
it consists of all \(Q\)-voters, the appropriate \(R\)-voters selected by the
set-system intersection, and one compatible block of \(U\)-voters from each
clause.  The value of \(\lambda\) is chosen so that such a canonical group has
exactly \(n\lambda\) voters.

The base satisfactions are calibrated so that every canonical
\((\lambda,n)\)-group has total base satisfaction
$n^2\lambda-\chi$,
whereas AJR requires total satisfaction at least
$n\cdot |H| = n^2\lambda$
for a canonical group \(H\) of size \(n\lambda\).  Hence, every canonical group
is short by exactly \(\chi\) units of satisfaction, and it must obtain these
\(\chi\) units from the actual committee.  In a yes-instance of
\sat, there exists an assignment \(x^*\) such
that \(\varphi(x^*,y)\) is false for every assignment \(y\).  We select the
blockers corresponding to \(x^*\).  Now consider any canonical group; it
corresponds to some satisfying full assignment \((\xi,\eta)\).  Since
\(\varphi(x^*,y)\) is false for every \(y\), we must have \(\xi\neq x^*\).
Choose an existential variable on which \(\xi\) and \(x^*\) differ.  This
variable appears in some clause, and the compatible \(U\)-voters for that clause
approve the selected opposite blocker.  Therefore, the group receives the missing
\(\chi\) units of satisfaction, and its average satisfaction reaches \(n\).

Conversely, suppose the original formula is a no-instance, so that
\(\forall x\exists y\,\varphi(x,y)\) holds.  Any committee satisfying the forced
groups must select all dummy candidates and exactly one blocker from each
existential pair \(B_i\), thereby encoding some assignment \(x^*\).  By the
no-instance assumption, there is an assignment \(y^*\) such that
\(\varphi(x^*,y^*)\) is true.  Consider the canonical group corresponding to
the satisfying full assignment \((x^*,y^*)\).  The \(Q\)- and \(R\)-voters in
this group approve only assignment candidates, none of which are selected.  The
compatible \(U\)-voters approve only blockers opposite to those selected by the
committee.  Thus this canonical group receives no actual satisfaction from the
committee, and its total satisfaction remains
$n^2\lambda-\chi<n^2\lambda$.
Its average satisfaction is therefore strictly less than \(n\), violating
\lamAJRB.  This yields the desired equivalence between the
\sat instance and the constructed
\lamAJRBExist instance.

\paragraph{A working example.}
Consider an instance with one existential variable $x_1$ and two universal
variables $y_1,y_2$, where
\[
\begin{aligned}
\varphi(x_1,y_1,y_2)
={}&(x_1\vee y_1\vee y_2)
\wedge(x_1\vee y_1\vee\neg y_2)\\
&{}\wedge(x_1\vee\neg y_1\vee y_2)
\wedge(x_1\vee\neg y_1\vee\neg y_2).
\end{aligned}
\]
This proper 3-CNF formula is equivalent to $x_1$: when $x_1=T$, every
clause is satisfied; when $x_1=F$, every assignment to $(y_1,y_2)$ makes
exactly one clause false. Thus, the \sat instance
$\exists x_1\,\forall y_1,y_2\,\neg\varphi(x_1,y_1,y_2)$ is a
\emph{yes-instance}, witnessed by $x_1=F$.
The construction has $n_1=1$, $n_2=2$, and
\[
    n=3,\qquad m=4,\qquad \chi=3,\qquad
    \gamma=85,\qquad \lambda=175,\qquad k=4.
\]
Here $z_1=x_1$, $z_2=y_1$, and $z_3=y_2$.

The voter blocks $D_1,D_2,D_3$ force the selection of all three dummy
candidates in $\varepsilon=\{\varepsilon_1,\varepsilon_2,\varepsilon_3\}$,
and $P_1$ forces at least one blocker from $B_1=\{b_1^T,b_1^F\}$.
Indeed, each of these blocks is a $(\lambda,1)$-group consisting of
$175$ voters with zero base satisfaction.
Since $k=4$, every \lamAJRB committee must therefore be one of
\[
    W_F=\varepsilon\cup\{b_1^F\}
    \qquad\text{or}\qquad
    W_T=\varepsilon\cup\{b_1^T\}.
\]
In particular, neither committee selects any assignment candidate.

To identify a concrete cohesive group, consider the satisfying full
assignment $\tau=(T,T,T)$ to $(z_1,z_2,z_3)$, and let $H$ consist of all
voters approving the three assignment candidates $a_1^T,a_2^T,a_3^T$.
By Lemma~\ref{lem:setsystem}, the intersection $F_1\cap F_2\cap F_3$
has $n(n-1)=6$ elements, each contributing $\gamma=85$ voters from $R$.
Each of the four clauses contributes its $\chi=3$ voters associated with
$\tau$, and all three voters in $Q$ belong to $H$.
Consequently,
\[
    |H\cap R|=510,\qquad |H\cap U|=12,\qquad |H\cap Q|=3,
\]
and $|H|=525=3\lambda$.
Since these voters commonly approve $a_1^T,a_2^T,a_3^T$, the set $H$
is a canonical $(\lambda,3)$-group.

By~\eqref{eqn:AJR-base-satisfaction}, the voters in $H\cap R$ have base
satisfaction $3$, those in $H\cap U$ have base satisfaction $2$, and
\[
    s(q_1)=s(q_2)=3,\qquad s(q_3)=n+m\chi-\chi=12.
\]
Thus,
\[
    \sum_{v\in H}s(v)
    =510\cdot3+12\cdot2+3+3+12
    =1572
    =3|H|-\chi.
\]
The group requires total satisfaction at least $3|H|=1575$, so it is
short by exactly $\chi=3$ units before receiving any satisfaction from
the committee.

Now consider $W_F$, which encodes the witnessing assignment $x_1=F$.
Every clause contains $x_1$, and the assignment defining $H$ has
$x_1=T$. Therefore, all $12$ voters in $H\cap U$ approve the opposite
blocker $b_1^F\in W_F$. The voters in $H\cap R$ and $H\cap Q$ approve
only assignment candidates, so they receive no satisfaction from $W_F$.
It follows that
\[
    S_{W_F}(H)+\sum_{v\in H}s(v)
    =12+1572
    =1584
    >1575=3|H|.
\]
The same calculation applies to the canonical group for every satisfying
full assignment: each such assignment has $x_1=T$, and its $12$
compatible clause voters all approve $b_1^F$.
As established in the characterization below, these are all the
$(\lambda,3)$-groups, and no $(\lambda,\ell)$-group exists for $\ell>3$.
The lower levels are also satisfied: every voter has total satisfaction
at least $1$ under $W_F$, and every $(\lambda,2)$-group excludes
$P\cup D$, so all its members have base satisfaction at least $2$.
Hence, $W_F$ satisfies \lamAJRB.

In contrast, $W_T$ fails on the specific group $H$ above.
Its clause voters approve $b_1^F$, not $b_1^T$, while its remaining
voters approve no blocker or dummy candidate. Thus $S_{W_T}(H)=0$, and
\[
    \frac{S_{W_T}(H)+\sum_{v\in H}s(v)}{|H|}
    =\frac{1572}{525}
    <3.
\]
Consequently, $W_T$ violates \lamAJRB, and $W_F$ is the unique
\lamAJRB committee in this example. This illustrates the distinct
roles of the two candidate families: the selected blocker encodes the
existential assignment, whereas the common assignment candidates
identify the cohesive group that tests it.

Finally, Lemmas~\ref{lem:lamAJRBtoLamAJR} and~\ref{lem:lamAJRtoAJR}
replace the base satisfactions by forced approved candidates and adjust
the Hare quota, respectively, while preserving committee existence.
Thus, the example also yields a yes-instance of \AJRE.

\paragraph{Formal proof for the validity of the construction.}
We first characterize the possible $(\lambda,\ell)$-groups.

\begin{itemize}
    \item For $\ell=1$, every $P_i$ and $D_j$ forms a $(\lambda,1)$-group.
    \item For every $\ell>1$, no $(\lambda,\ell)$-group contains voters from $P$ or $D$.
    \item For every $\ell>n$, no $(\lambda,\ell)$-group exists.
    \item Every $(\lambda,n)$-group is canonical: it contains
    \[
        n(n-1)\gamma\text{ voters from }R,\qquad m\chi\text{ voters from }U,\qquad n\text{ voters from }Q,
    \]
    and corresponds to a satisfying full assignment of $\varphi$.
\end{itemize}

The first case is immediate.
For the second case, fix $i\le n_1$ and suppose a group $G$ contains a voter from $P_i$.
Since each voter in $P_i$ approves only $B_i$, a group containing such a voter can have at least two common approved candidates only if every voter in the group approves both blockers in $B_i$.
Only voters in $P_i$ approve both blockers in $B_i$: voters in $U$ approve at most one blocker from any $B_i$, and voters in $Q$, $R$, and $D$ approve no blockers in $B_i$.
Thus, such a group has size at most $|P_i|=\lambda<2\lambda$, so it cannot be a $(\lambda,\ell)$-group for any $\ell>1$.
The same reasoning applies to $D$, where each voter is a singleton-approval voter.

Now consider a $(\lambda,\ell)$-group $G$ with $\ell\ge n$.
By the previous paragraph, $G\subseteq U\cup Q\cup R$.
Since
\[
    |U\cup Q|=7m\chi+n=\gamma-1+n<n\lambda,
\]
we have $G\cap R\neq\emptyset$.
Thus, every candidate commonly approved by $G$ is an assignment candidate in $A$.

We first show that the common assignment candidates of $G$ cannot contain a non-canonical choice of $n$ candidates.
Suppose, toward a contradiction, that among the common assignment candidates there are $n$ candidates whose corresponding sets in $\mathcal F$ do not contain exactly one member from each pair $\mathcal F_i=\{F_i,\bar F_i\}$.
Every voter in $G\cap R$ approves these $n$ candidates.
Hence, if a voter from the block $R_z$ belongs to $G\cap R$, then the element $z$ lies in the intersection of the corresponding $n$ sets.
By Lemma~\ref{lem:setsystem}, this intersection has size at most $n(n-1)-1$.
Therefore, the voters in $G\cap R$ can come from at most $n(n-1)-1$ blocks $R_z$, and so
\[
    |G\cap R|\le (n(n-1)-1)\gamma.
\]
Using $|U\cup Q|=7m\chi+n=\gamma-1+n$, we obtain
\[
    |G|\le |U\cup Q|+(n(n-1)-1)\gamma
    = n(n-1)\gamma+n-1
    < m\chi+n(n-1)\gamma+n=n\lambda,
\]
contradicting $|G|\ge \ell\lambda\ge n\lambda$.

It follows that every $n$ common assignment candidates form a canonical choice, i.e., exactly one candidate from each pair $A_i$.
In particular, $G$ has at most $n$ common candidates: if it had more than $n$, then some pair $A_i$ would contribute two of them, and one could choose a non-canonical $n$-subset.
Therefore no $(\lambda,\ell)$-group exists for $\ell>n$.

It remains to characterize $(\lambda,n)$-groups.
Let $H$ be such a group.
The preceding paragraph implies that its common candidates are exactly $n$ assignment candidates, one from each pair $A_i$.
These candidates define a full assignment $\tau$.
By Lemma~\ref{lem:setsystem}, the corresponding $n$ sets have intersection size exactly $n(n-1)$.
Therefore at most $n(n-1)$ blocks $R_z$ can intersect $H$, and hence
\[
    |H\cap R|\le n(n-1)\gamma.
\]
For a fixed clause block $U_j$, only voters whose partial assignment is compatible with $\tau$ can belong to $H$.
If the clause is satisfied by $\tau$, there are exactly $\chi$ identical compatible voters in $U_j$; otherwise there are none.
Thus each clause block contributes at most $\chi$ voters, so
\[
    |H\cap U|\le m\chi.
\]
Also, $|H\cap Q|\le n$.
Combining these upper bounds gives
\[
    |H|
    \le n(n-1)\gamma+m\chi+n
    = n\lambda.
\]
Since $H$ is a $(\lambda,n)$-group, equality must hold throughout.
Thus $H$ contains exactly $n(n-1)\gamma$ voters from $R$, exactly $m\chi$ voters from $U$, and all $n$ voters from $Q$.
Moreover, every clause block contributes $\chi$ compatible voters, so $\tau$ satisfies every clause of $\varphi$.
Consequently, every $(\lambda,n)$-group is canonical with size exactly $n\lambda$.

Conversely, every satisfying full assignment $\tau$ gives such a canonical $(\lambda,n)$-group.
Choose the $n(n-1)\gamma$ voters from the $R_z$ blocks whose elements lie in the intersection of the $n$ set-system members corresponding to $\tau$.
For each clause block $U_j$, choose the $\chi$ voters whose satisfying partial assignment is compatible with $\tau$, and include all voters in $Q$.
The selected voters have size
\[
    n(n-1)\gamma+m\chi+n=n\lambda
\]
and commonly approve the $n$ assignment candidates specified by $\tau$.
Thus they form a canonical $(\lambda,n)$-group.

\begin{proposition}\label{prop:base_sat}
    For every canonical $(\lambda,n)$-group $H$,
    \[
        \sum_{v\in H}s(v)=n^2\lambda-\chi.
    \]
\end{proposition}
\begin{proof}
    A canonical group has $n(n-1)\gamma$ voters from $R$, $m\chi$ voters from $U$, and all voters from $Q$.
    Hence,
    \begin{align*}
        \sum_{v\in H}s(v)
        &= n\cdot n(n-1)\gamma+(n-1)m\chi+\big((n-1)n+(n+m\chi-\chi)\big)\\
        &= n^2\lambda-\chi.\qedhere
    \end{align*}
\end{proof}

\begin{lemma}[Completeness]
If the \sat instance $\varphi$ is a yes-instance, then the constructed \lamAJRBExist instance admits a \lamAJRB committee of size $k$.
\end{lemma}

\begin{proof}
    Suppose the \sat instance is a yes instance.
    Then there exists an assignment $\mathbf{x}^*=(x_1^*,\ldots,x_{n_1}^*)$ such that for every assignment $\mathbf{y}$ the formula $\varphi(\mathbf{x}^*,\mathbf{y})$ is false.
    Define
    \[
        W=\{b_i^{x_i^*}:i=1,\ldots,n_1\}\cup\varepsilon.
    \]
    Then $|W|=k$.

    All $(\lambda,1)$-groups are satisfied: voters in $P\cup D$ approve at least one selected candidate, and every other voter has base satisfaction at least $n-1\ge 1$.
    For $1<\ell<n$, no $(\lambda,\ell)$-group contains voters from $P$ or $D$, and every remaining voter has base satisfaction at least $n-1\ge \ell$.

    Finally, let $H$ be a canonical $(\lambda,n)$-group, and let $(\xi,\eta)$ be the satisfying full assignment to which it corresponds.
    Since $\varphi(\mathbf{x}^*,\mathbf{y})$ is false for every $\mathbf{y}$, we have $\xi\ne\mathbf{x}^*$.
    Choose an existential variable $x_i$ on which $\xi$ and $\mathbf{x}^*$ differ.
    By assumption, $x_i$ appears in some clause.
    In a clause containing $x_i$, the $\chi$ compatible voters in $H\cap U$ approve $b_i^{\overline{\xi_i}}=b_i^{x_i^*}\in W$.
    Thus, $S_W(H)\ge \chi$, and Proposition~\ref{prop:base_sat} gives
\[
        S_W(H)+\sum_{v\in H}s(v)\ge \chi+(n^2\lambda-\chi)=n^2\lambda.
\]
    Since $|H|=n\lambda$, the average satisfaction of $H$ is at least $n$.
    Therefore $W$ satisfies \lamAJRB.
\end{proof}

\begin{lemma}[Soundness]
    If the \sat instance $\varphi$ is a no-instance, then the constructed \lamAJRBExist instance does not admit a \lamAJRB committee of size $k$.
\end{lemma}
\begin{proof}
    Suppose $\forall \mathbf{x}\exists \mathbf{y}\,\varphi(\mathbf{x},\mathbf{y})$ holds, and let $W$ be any size-$k$ committee satisfying \lamAJRB.
    The groups $D_j$ force all dummy candidates into $W$, and the groups $P_i$ force at least one blocker from each $B_i$ into $W$.
    Since $k=n_1+n$, the committee contains exactly all dummy candidates and exactly one blocker from each $B_i$.
    Write
    \[
        W\cap B_i=\{b_i^{t_i}\}
    \]
    and let $\mathbf{x}^*=(t_1,\ldots,t_{n_1})$.
    By assumption, choose $\mathbf{y}^*$ such that $\varphi(\mathbf{x}^*,\mathbf{y}^*)=1$.
    Let $H^*$ be the canonical $(\lambda,n)$-group corresponding to $(\mathbf{x}^*,\mathbf{y}^*)$.

    Voters in $Q$ and $R$ approve only assignment candidates.
    The compatible voters in $U$ approve, for existential variables in their clauses, only blockers of the form $b_i^{\overline{x_i^*}}$, while $W$ contains only blockers $b_i^{x_i^*}$ and dummy candidates.
    Therefore $S_W(H^*)=0$.
    By Proposition~\ref{prop:base_sat},
    \[
        S_W(H^*)+\sum_{v\in H^*}s(v)
        =n^2\lambda-\chi<n^2\lambda.
    \]
    Since $|H^*|=n\lambda$, the average satisfaction of $H^*$ is strictly less than $n$, contradicting \lamAJRB.
\end{proof}
This completes the reduction and establishes the \(\SigmaTwo\)-hardness of \lamAJRBExist.
Hence, \lamAJRBExist is \(\SigmaTwo\)-hard.

\subsection{Reduction from \lamAJRBExist to \AJRE}
\label{sec:AJR_basetoAJR}
We prove the following two lemmas in this section, which will conclude Theorem~\ref{thm:AJR}.
The first reduction removes base satisfactions by replacing each unit of base satisfaction with a forced approved candidate.
\begin{lemma}\label{lem:lamAJRBtoLamAJR}
    \lamAJRBExist is Karp reducible to \lamAJRExist.
\end{lemma}

\begin{proof}
    Let $(E,\lambda,s)$ be an instance of \lamAJRBExist, where $E=(N,C,\mathcal{A},k)$. 
    By the definition of \lamAJRBExist, $\lambda\le |N|$ and $\sum_{i\in N}s(i)$ is polynomially bounded in the input size; in particular, every value $s(i)$ is polynomially bounded.
    Hence the voters and candidates added below form a polynomial-size instance.
    We construct an instance $(\mathbb{E},\lambda')$ of \lamAJRExist, with $\mathbb{E}=(N',C',\mathcal{A}',k')$, as follows. 
    Set
    \[
    k' = k + \sum_{i\in N} s(i), \qquad \lambda' = \lambda.
    \]
    For each voter $i\in N$ with $s(i)>0$ and for each $j\in\{1,\ldots,s(i)\}$, create a set $N_i^j=\{n^{(i,j)}_1,\ldots,n^{(i,j)}_{\lambda}\}$ of $\lambda$ new voters. 
    Let
    \[
    N' = N \cup \bigcup_{\substack{i\in N,\,s(i)>0\\ j\in[1,s(i)]}} N_i^j.
    \]
    Since \(\lambda'=\lambda\le |N|\le |N'|\), the output satisfies the input restriction of \lamAJRExist.
    For each $i\in N$ with $s(i)>0$, introduce a candidate set $C_i=\{c^i_1,\ldots,c^i_{s(i)}\}$; if $s(i)=0$, let $C_i=\emptyset$. 
    Then set
    \[
    C' = C \cup \bigcup_{i\in N} C_i.
    \]

    Define the approval sets as follows. For each $i\in N$, let $A'_i = A_i \cup C_i$. For each $n^{(i,j)}_k\in N_i^j$, set $A'_{n^{(i,j)}_k}=\{c^i_j\}$. 

    For every positive integer $\ell$ and every $G\subseteq N$, $G$ is a $(\lambda,\ell)$-group in $E$ if and only if it is a $(\lambda',\ell)$-group in $\mathbb{E}$.
    Indeed, an old-only $(\lambda',\ell)$-group has size at least $\lambda>2$, so it contains at least two old voters; the added candidate sets $C_i$ are pairwise disjoint across old voters, and therefore they do not change the common approved candidate set of any old-only group.
    Any other $(\lambda',\ell)$-group in $\mathbb{E}$ must include some voters from $\bigcup_{i,j} N_i^j$, but since each such voter approves exactly one candidate, the common approved candidate set of these groups has size at most $1$. 
    Hence, they are all $(\lambda',1)$-groups. 

    Suppose $W$ is a committee of size $k$ that satisfies \lamAJRB for $(E,\lambda,s)$. Define
    \[
    W' = W \cup \bigcup_{i\in N} C_i.
    \]
    Then $W'$ is a committee of size $k'$ for $\mathbb{E}$. 
    Consider any $(\lambda',\ell)$-group $G'$ in $\mathbb{E}$. 
    If $G'$ contains voters not in $N$, then its common approval set has size at most $1$, so $\ell=1$.
    Since the corresponding base candidate is selected, every voter in $G'$ approves a selected common candidate, and the average satisfaction of $G'$ is at least $1$.
    If $G'\subseteq N$, then $G'$ is an old $(\lambda,\ell)$-group.
    Since $W$ satisfies \lamAJRB,
    \[
        S_W(G')+\sum_{i\in G'}s(i)\ge \ell |G'|.
    \]
    The selected base candidates contribute exactly $\sum_{i\in G'}s(i)$, so
    \[
        S_{W'}(G')=S_W(G')+\sum_{i\in G'}s(i)\ge \ell |G'|.
    \]
    Hence, $W'$ satisfies $\lambda'$-AJR for $(\mathbb{E},\lambda')$.

    Conversely, suppose $W'$ is a committee of size $k'$ that satisfies $\lambda'$-AJR for $(\mathbb{E},\lambda')$. 
    Each candidate $c^i_j\in C_i$ is approved by exactly $\lambda$ voters in $N_i^j$, so $W'$ must contain all of $\bigcup_{i\in N} C_i$. 
    Thus, $|W'\cap C|=k$. Let $W=W'\cap C$. 
    We claim $W$ satisfies \lamAJRB for $(E,\lambda,s)$.
    Indeed, for any $(\lambda,\ell)$-group $G$ in $E$, the corresponding group in $\mathbb{E}$ satisfies
    \[
        S_{W'}(G)\ge \ell |G|.
    \]
    Since all base candidates are selected,
    \[
        S_{W'}(G)=S_W(G)+\sum_{i\in G}s(i).
    \]
    Therefore
    \[
        S_W(G)+\sum_{i\in G}s(i)\ge \ell |G|,
    \]
    showing that $W$ satisfies \lamAJRB.

    This establishes equivalence between $(E,\lambda,s)$ and $(\mathbb{E},\lambda')$, completing the proof.
\end{proof}

The second reduction adjusts the Hare quota to $\lambda$ by adding forced candidates and carefully matched new voters.
\begin{lemma}\label{lem:lamAJRtoAJR}
    \lamAJRExist is Karp reducible to \AJRE.
\end{lemma}

\begin{proof}
Let $(E,\lambda)$ be an instance of \lamAJRExist, where $E=(N,C,A,k)$ and $\lambda\in\mathbb{Z}_{>2}$, and let $n=|N|$.

\medskip
\noindent\textbf{Case 1: $\lambda k=n$.}
Then $\lambda=n/k$, so \lamAJR coincides with AJR. Hence the claim is immediate.

\medskip
\noindent\textbf{Case 2: $\lambda k>n$.}
Let $\varepsilon=\lambda k-n$, and define an \AJRE instance $E'=(N',C',A',k')$ by setting
$N'=N\cup Z$ with $|Z|=\varepsilon$, $C'=C$, and $k'=k$, where every $z\in Z$ approves no candidates, i.e., $A'_i=A_i$ for $i\in N$ and $A'_z=\emptyset$ for $z\in Z$.

Then $|N'|=n+\varepsilon=n+(\lambda k-n)=\lambda k=\lambda k'$, so the Hare quota in $E'$ is exactly $\lambda$.
Since the added voters approve nothing, they belong to no cohesive group. Thus a set $N_0\subseteq N$ is a $(\lambda,\ell)$-group in $E$ if and only if it is an $\ell$-cohesive group in $E'$. Therefore, a committee of size $k$ satisfies \lamAJR in $E$ if and only if it satisfies AJR in $E'$.

\medskip
\noindent\textbf{Case 3: $\lambda k<n$.}
Let $d=n-\lambda k>0$ and $\xi=(\lambda-1)d$. We construct an \AJRE instance $E'=(N',C',A',k')$.

Add $\xi$ new candidates $Y=\{y_1,\dots,y_\xi\}$, partitioned into $d$ pairwise disjoint blocks $B_1,\dots,B_d$ of size $\lambda-1$ each. For every $j\in[\xi]$, add a set $X_j=\{x_{j,1},\dots,x_{j,\lambda-1}\}$ of $\lambda-1$ voters, all approving only $y_j$. For every $t\in[d]$, add one voter $q_t$ approving exactly the block $B_t$. If $\lambda>3$, add $(\lambda-3)d$ dummy voters $R$, each approving nothing; if $\lambda=3$, let $R=\emptyset$.

Set
\[
N'=N\cup \bigcup_{j=1}^{\xi} X_j \cup \{q_1,\dots,q_d\}\cup R,\qquad
C'=C\cup Y,\qquad
k'=k+\xi.
\]
The approval sets are given by $A'_i=A_i$ for $i\in N$, $A'_u=\{y_j\}$ for $u\in X_j$, $A'_{q_t}=B_t$ for $t\in[d]$, and $A'_r=\emptyset$ for $r\in R$.

We first compute the Hare quota in $E'$. The number of added voters is
$\xi(\lambda-1)+d+(\lambda-3)d=\xi(\lambda-1)+(\lambda-2)d$.
Using $\xi=(\lambda-1)d$ and $n=\lambda k+d$, we obtain
\begin{align*}
|N'|
&=n+\xi(\lambda-1)+(\lambda-2)d \\
&=(\lambda k+d)+(\lambda-1)^2d+(\lambda-2)d \\
&=\lambda k+\lambda(\lambda-1)d \\
&=\lambda\bigl(k+(\lambda-1)d\bigr)
=\lambda(k+\xi)
=\lambda k'.
\end{align*}
Hence the Hare quota in $E'$ is exactly $\lambda$.

Next we characterize the cohesive groups containing new voters. Old voters approve only candidates in $C$, while new non-dummy voters approve only candidates in $Y$, so no cohesive group contains both old and new voters. Voters in $R$ approve nothing and hence belong to no cohesive group. Now for each $j\in[\xi]$, let $b(j)\in[d]$ be the unique index such that $y_j\in B_{b(j)}$. Among the new voters, the only voters approving $y_j$ are the $\lambda-1$ voters in $X_j$ and the single voter $q_{b(j)}$. Therefore any cohesive group containing a new voter is contained in
\[
G_j:=X_j\cup\{q_{b(j)}\}.
\]
Moreover, $G_j$ has common approved candidate set $\{y_j\}$ and size $|G_j|=(\lambda-1)+1=\lambda$, so $G_j$ is $1$-cohesive. Since the common intersection is always a singleton, no cohesive group containing new voters is $\ell$-cohesive for any $\ell\ge 2$.

We now prove the equivalence.

\medskip
\noindent\emph{Soundness.}
Let $W'\subseteq C'$ be a committee of size $k'$ satisfying AJR in $E'$. We first show that $Y\subseteq W'$. Suppose $y_j\notin W'$ for some $j\in[\xi]$. Then $G_j$ is a $1$-cohesive group. Every voter in $X_j$ has satisfaction $0$, while $q_{b(j)}$ approves the $\lambda-1$ candidates in $B_{b(j)}$, one of which is $y_j$. Hence if $y_j\notin W'$, then $q_{b(j)}$ approves at most $\lambda-2$ members of $W'$, and so
\[
S_{W'}(G_j)\le \lambda-2
\qquad\text{and}\qquad
\frac{S_{W'}(G_j)}{|G_j|}\le \frac{\lambda-2}{\lambda}<1,
\]
contradicting AJR. Thus $Y\subseteq W'$.

Now let $W=W'\cap C$. Since $|W'|=k'=k+\xi$ and $|Y|=\xi$, we have $|W|=k$.

We claim that $W$ satisfies \lamAJR in $E$.
Let $N_0\subseteq N$ be any $(\lambda,\ell)$-group in $E$.
The group-size requirement gives $\ell\lambda\le |N_0|\le |N|\le |N'|=\lambda k'$, and hence $\ell\le k'$.
Since old approval sets are unchanged and $|N'|/k'=\lambda$, $N_0$ is an $\ell$-cohesive group in $E'$.
Also, old voters approve no candidates in $Y$, so $S_W(N_0)=S_{W'}(N_0)$.
As $W'$ satisfies AJR in $E'$, we get
\[
\frac{S_W(N_0)}{|N_0|}
=
\frac{S_{W'}(N_0)}{|N_0|}
\ge \ell.
\]
Hence $W$ satisfies \lamAJR in $E$.

\medskip
\noindent\emph{Completeness.}
Let $W\subseteq C$ be a committee of size $k$ satisfying \lamAJR in $E$, and define $W'=W\cup Y$. Then $|W'|=k+\xi=k'$.

We show that $W'$ satisfies AJR in $E'$.
Let $N_0\subseteq N$ be an old $\ell$-cohesive group in $E'$.
Since old approval sets are unchanged and $|N'|/k'=\lambda$, $N_0$ is a $(\lambda,\ell)$-group in $E$.
Since $W$ satisfies \lamAJR and every voter has satisfaction at most $k$, such a group necessarily has $\ell\le k$.
Old voters approve no candidates in $Y$, so we obtain
\[
\frac{S_{W'}(N_0)}{|N_0|}
=
\frac{S_W(N_0)}{|N_0|}
\ge \ell.
\]

Second, let $N_0$ be a cohesive group containing new voters. By the above characterization, $N_0\subseteq G_j$ for some $j\in[\xi]$, and $y_j\in W'$. Thus every voter in $N_0$ approves at least one selected candidate, so $\frac{S_{W'}(N_0)}{|N_0|}\ge 1$. Since no such group is $\ell$-cohesive for any $\ell\ge 2$, this is exactly what AJR requires.

Therefore every cohesive group in $E'$ satisfies AJR under $W'$, so $W'$ is an AJR committee in $E'$.

\medskip
Thus, in every case, we obtain in polynomial time an instance $E'$ of AJR-E such that $E$ has a \lamAJR committee if and only if $E'$ has an AJR committee. This proves the claim.
\end{proof}

\iffalse
Combining Theorems~\ref{thm:lamAJRBtoLamAJR} and~\ref{thm:lamAJRtoAJR}, we have the following corollary.
\begin{corollary}
    \lamAJRBExist is polynomial-time Karp reducible to \AJRE, i.e., \lamAJRBExist\karpreduct\AJRE.
\end{corollary}

Hence, by the $\SigmaTwo$-completeness of \lamAJRBExist, the polynomial-time Karp reduction to \AJRE and lemma~\ref{lem:AJRmem}, we proved that \AJRE is $\SigmaTwo$-complete.

\fi

By Lemma~\ref{lem:AJR}, \lamAJRBExist is $\SigmaTwo$-hard.
Lemma~\ref{lem:lamAJRBtoLamAJR} and Lemma~\ref{lem:lamAJRtoAJR} give a polynomial-time Karp reduction from \lamAJRBExist to \AJRE.
Together with Lemma~\ref{lem:AJRmem}, this proves Theorem~\ref{thm:AJR}.

\section{Complexity Result for SJR}
\label{sec:Main Result}
In this section, we prove the following result.
\begin{theorem}\label{thm:sjr}
    \SJRE is $\ThetaTwo$-complete.
\end{theorem}

We will prove that \SJRE is in $\ThetaTwo$ in Sect.~\ref{sec:SJR_membership}.
The remaining sections aim to prove the $\ThetaTwo$-hardness of \SJRE.
To prove the $\ThetaTwo$-hardness, we introduce two intermediate problems \lamSJRExist and \lamSJRBExist, and we prove the following chain of Karp reductions:
\[
\text{\vcmember}\leq_K\text{\lamSJRBExist}\leq_K\text{\lamSJRExist}\leq_K\text{\SJRE}.
\]
We will prove $\text{\vcmember}\leq_K\text{\lamSJRBExist}$ in Sect.~\ref{sec:SJR_vcmembertobase} and $\text{\lamSJRBExist}\leq_K\text{\lamSJRExist}\leq_K\text{\SJRE}$ in Sect.~\ref{sec:SJR_basetoSJR}.
The purpose of defining \lamSJRExist and \lamSJRBExist is the same as before: the first part $\text{\vcmember}\leq_K\text{\lamSJRBExist}$ contains the essential ideas of the hardness of the problem, and the introduction of the two intermediate problems is mainly for the enhancement of readability.

\iffalse
We will formally define \lamSJRExist and \lamSJRBExist soon.
Before this, we describe some high-level ideas on these two decision problems and the intuitive reason why we introduce them.
\lamSJRExist is a variant of \SJRE where $\ell$-cohesive groups in the definition of SJR are replaced by $(\lambda,\ell)$-groups (Definition~\ref{def:lambdaellgroup}).
In other words, instead of setting the size of the $1$-cohesive group to $n/k$, we allow a customized size of $\lambda$ where $\lambda$ is specified as an input.
\lamSJRBExist is a variant where, in addition to the customized cohesive group sizes, each voter $i$ also has a customized ``base satisfaction'' $s(i)$ such that $i$ has a base satisfaction of $s(i)$ when she approves no candidate in the committee and $i$'s satisfaction is redefined to $S_W(i)+s(i)$ for committee $W$.
The first part of the reduction chain, $\text{\vcmember}\leq_K\text{\lamSJRBExist}$ contains the essential ideas on the hardness of our problem.
For the second part of the reduction chain $\text{\lamSJRBExist}\leq_K\text{\lamSJRExist}\leq_K\text{\SJRE}$, it is achieved by carefully adding dummy voters and candidates.
We use these two intermediate problems to increase the accessibility of our proof.
\fi

%\paragraph{Two Intermediate Problems.}
We first introduce the modified version of SJR, called $\lambda$-\emph{SJR}.
\begin{definition}[$\lambda$-Strong Justified Representation (\lamSJR)]
    A committee \(W\subseteq C\) of size \(k\) satisfies $\lambda$-\emph{Strong Justified Representation} if for every positive integer $\ell$ there exists \emph{no} \((\lambda,\ell)\)-cohesive group \(N'\subseteq N\) such that $S_W(i)<\ell$ for some $i\in N'$.
    %\[
    %    \exists i\in N',\;S_W(i)<\ell.
    %\]
    Equivalently, each member of a \((\lambda,\ell)\)-cohesive group has satisfaction at least \(\ell\).
\end{definition}
Next, we define a variant of \lamSJRExist where each voter $i$ has a ``base satisfaction'' $s(i)$ (given as an input parameter) such that her satisfaction for committee $W$ is redefined to $S_W(i)+s(i)$.

\begin{definition}[\lamSJR Committee with base satisfactions (\lamSJRB)]
    Given an election $\mathbb{E}=(N,C,\mathcal{A},k)$, and a base satisfaction function $s:N\to\mathbb{Z}_{\geq0}$, 
    a committee $W\subseteq C$ with $|W|=k$ satisfies \(\lambda\)-\emph{Strong Justified Representation with base satisfactions} (\lamSJRB) 
    if for every positive integer $\ell$ there exists \emph{no} $(\lambda,\ell)$-group $N'\subseteq N$ such that $S_W(i)+s(i)< \ell$ for some $i\in N'$.
    %\[
    %    \exists i\in N',\;S_W(i)+s(i)< \ell.
    %\]
    Equivalently, the minimum satisfaction of each \((\lambda,\ell)\)-group is at least \(\ell\), when each voter $i$ has a base satisfaction $s(i)$.
\end{definition}

As in the AJR section, the auxiliary \(\lambda\)-SJR problems below require the input parameter \(\lambda\) to satisfy \(2<\lambda\le |N|\).
An input to \lamSJRBExist must additionally satisfy \(\sum_{i\in N}s(i)\le |\mathbb{E}|^{O(1)}\).
These input restrictions are part of the definitions below.

Then, we define the corresponding decision problems.
\begin{definition}[\lamSJR Committee Existence (\lamSJRExist)]
    Given an election $\mathbb{E}=(N,C,\mathcal{A},k)$ and an integer parameter \(2<\lambda\le |N|\),
    decide whether there exists a committee $W\subseteq C$ with $|W|=k$ that satisfies \lamSJR. 
    The output is “Yes” if such a committee exists, and “No” otherwise.
\end{definition}
\begin{definition}[\lamSJRB Committee Existence (\lamSJRBExist)]
    Given an election $\mathbb{E}=(N,C,\mathcal{A},k)$, a base satisfaction function $s:N\to\mathbb{Z}_{\geq0}$ with polynomially bounded total value, and an integer parameter \(2<\lambda\le |N|\),
    decide whether there exists a committee $W\subseteq C$ with $|W|=k$ that satisfies \lamSJRB. 
    The output is “Yes” if such a committee exists, and “No” otherwise.
\end{definition}

\subsection{Membership in $\ThetaTwo$}\label{sec:SJR_membership}
In this section, we prove that \SJRE is in $\ThetaTwo$.
It suffices to show that \SJRE is in $\Poly_{\parallel+1}^\NP$, and Theorem~\ref{thm:complexity} implies the containment in $\ThetaTwo$.
Let \(\mathbb{E}\) be an instance of \SJRE with \(\mathbb{E}=(N,C,\mathcal{A},k)\) and \(N=\{1,\dots,n\}\).
%A committee \(W\subseteq C\) of size \(k\) satisfies SJR if, for every voter \(i\in N\) and every \(\ell\ge1\),
%\[
%    \text{whenever } i \text{ belongs to some \(\ell\)-group, then } |A_i\cap W|\ge\ell.
%\] 
%We give a deterministic polynomial-time algorithm with \(O(\log n)\) calls to a \(\NP\)-oracle.
For each voter \(i\), define
\[
    a_i \;=\; \max\left(\{0\}\cup\left\{\ell : \exists\,N'\subseteq N,\;i\in N',\;
    |N'|\ge\tfrac{\ell\,n}{k},\;
    \left|\bigcap_{j\in N'}A_j\right|\ge\ell\right\}\right).
\]
That is, $a_i$ is the maximum $\ell$ such that voter $i$ is contained in some $\ell$-cohesive group.
Testing whether voter \(i\) belongs to an \(\ell\)-cohesive group is in \(\NP\): we can guess the set \(N'\) and verify in polynomial time
$
    |N'|\ge \tfrac{\ell\,n}{k}
    \text{ and }
    \Bigl|\bigcap_{j\in N'}A_j\Bigr|\ge \ell.
$
Since \(\ell\in\{1,\dots,k\}\), each $a_i$ can be computed by $k$ non-adaptive NP queries: just check if $i$ is in an $\ell$-cohesive group for each $\ell=1,\ldots,k$.
Therefore, we can compute the vector $(a_1,\ldots,a_n)$ by $nk$ non-adaptive NP queries.
%each \(a_i\) can be determined by parallel queries to a \(\NP\)-oracle, and repeating for all \(i\in N\) remains in \(\Poly^{\NP}_{||}\).

Once the thresholds $a_i$ are fixed, it suffices to test whether there exists a committee
$W\subseteq C$ of size $k$ such that $|A_i\cap W|\ge a_i$ for all $i\in N$.
This feasibility check is in $\NP$, with $W$ as the certificate.

Moreover, the constraints $|A_i\cap W|\ge a_i$ for every $i$ capture SJR in a pointwise way:
If for some $i$ we have $|A_i\cap W|<a_i$, then by the definition of $a_i$ there exists an $a_i$-cohesive group $S\ni i$. Since $i$ already falls short of $a_i$, $W$ violates SJR for the group $S$.
Conversely, any SJR violation yields such an $\ell$ and a voter $i$ in the violating group with
$|A_i\cap W|<\ell\le a_i$.
The procedure makes \(nk\) parallel calls to an \(\NP\)-oracle to compute the thresholds and one additional call to check the existence of $W$.
The remaining computations can be done in polynomial time.
Hence, the problem lies in $\Poly_{\parallel+1}^\NP$ and thus is in \(\ThetaTwo\).

\subsection{Reduction from \vcmember to \lamSJRBExist}
\label{sec:SJR_vcmembertobase}
Here, we describe some high-level ideas. Consider a \vcmember instance $(G=(V,E),x)$.
To simulate the vertex cover problem with $G=(V,E)$, we construct the voter set $N_G$ and the candidate sets $C_G$, where voters correspond to the edges and candidates correspond to vertices.
By a careful construction, we can make sure that we must select the candidates in $C_G$ corresponding to a vertex cover to ensure SJR.
Moreover, we add a voter set $N_x$ for the special vertex $x$ in the \vcmember problem that approves the candidate in $C_G$ representing $x$; this ensures the candidate representing $x$ must be selected.

To ensure that the vertex cover is a minimum one, we need some other constructions to make sure no more than $k'$ candidates in $C_G$ can be selected for $k'$ being the size of the minimum vertex cover.
We use this simple and well-known fact: $S$ is a vertex cover of $G=(V,E)$ if and only if $V\setminus S$ is a clique in the complement graph $\overline{G}=(V,E^c)$.
Therefore, the size of the minimum vertex cover in $G$ plus the size of the maximum clique in $\overline{G}$ is always $|V|$.
For the complement-clique requirement, we use certificate gadgets instead of directly constructing approvals as on the cover-forcing side.
The construction contains one \emph{copy} for every vertex \(v_i\).
Copy \(i\) is a complete local copy of the vertex-cover candidates, clique-witness candidates, and voters, and it is used to measure the largest clique in \(\overline G\) that contains \(v_i\).
Inside each copy, we introduce one certificate \emph{level} \(q\) for every possible clique size \(q\in[m_v]\).
At level \(q\) of copy \(i\), we build a formula whose satisfying assignments are exactly the \(q\)-cliques in \(\overline{G}\) that contain \(v_i\).
If such a clique exists, the corresponding formula gadget creates a cohesive group that forces the clique-witness candidate \(w_q^i\).
Thus, in copy \(i\), every feasible committee must reserve at least \(\tau_i\) clique-witness candidates, where \(\tau_i\) is the largest size of a clique in \(\overline{G}\) containing \(v_i\).
We need copies for all vertices because we do not know in advance which vertex lies in a maximum clique of \(\overline G\), and the soundness argument uses a copy \(i\) with \(\tau_i=\alpha(G)\).

To make sure we need to select exactly $|V|$ active candidates in each copy, we add another voter set $N_f$ in each copy that approves all active candidates in that copy.
We also add forced dummy candidates that are selected in every feasible committee.
These dummy candidates leave exactly $|V|$ active seats in each copy.

We use the following standard certificate gadget, which is the $\chi=1$ specialization of the canonical-assignment construction in Sect.~\ref{sec:AJR_sattobase}.
\begin{lemma}[Formula verifier gadget]\label{lem:SJR-formula-gadget}
Let $\varphi$ be a proper 3-CNF formula with $n_0>2$ variables and $m_0$ clauses,
where $m_0/n_0$ is an integer. Let
\[
    \gamma=7m_0+1,\qquad
    \lambda=\frac{m_0}{n_0}+(n_0-1)\gamma+1.
\]
Using only the assignment-candidate part and the voter sets \(U,Q,R\) from the construction in the proof of Lemma~\ref{lem:AJR}, specialized to $\chi=1$ with $n=n_0$ and $m=m_0$, one can construct
a candidate set \(C^{\mathrm{asgn}}(\varphi)\) and a voter set $T(\varphi)$ such that, for every \(t\in T(\varphi)\), the notation \(A_t^{\mathrm{asgn}}\subseteq C^{\mathrm{asgn}}(\varphi)\) denotes the approval set of \(t\) inside this standalone certificate gadget:
\begin{enumerate}
    \item If $\varphi$ is satisfiable, then there exists
    $T'\subseteq T(\varphi)$ with
    \[
        |T'|=\lambda n_0
        \quad\text{and}\quad
        \left|\bigcap_{t\in T'}A_t^{\rm asgn}\right|=n_0.
    \]
    \item If there exists $T'\subseteq T(\varphi)$ with
    \[
        |T'|\ge \lambda n_0
        \quad\text{and}\quad
        \left|\bigcap_{t\in T'}A_t^{\rm asgn}\right|\ge n_0,
    \]
    then $\varphi$ is satisfiable.
    \item For every $r>n_0$, there is no $T'\subseteq T(\varphi)$ with
    \[
        |T'|\ge \lambda r
        \quad\text{and}\quad
        \left|\bigcap_{t\in T'}A_t^{\rm asgn}\right|\ge r.
    \]
\end{enumerate}
\end{lemma}
\begin{proof}
Take exactly the assignment candidates and the voter sets $U,Q,R$ from
the construction in the proof of Lemma~\ref{lem:AJR}, specialized to $\chi=1$ with $n=n_0$ and $m=m_0$.
Omit the blocker candidates, the voters $P,D$, and all dummy candidates used there; equivalently, remove blocker approvals from the $U$-voters.
The first property is the canonical satisfying-assignment group in the proof
of Lemma~\ref{lem:AJR}. The second and third properties are precisely the
characterization of large cohesive groups in the $U\cup Q\cup R$ part of
that proof under this $\chi=1$ specialization: no $(\lambda,r)$-group exists for $r>n_0$, and every
$(\lambda,n_0)$-group corresponds to a truth assignment satisfying
$\varphi$.
\end{proof}

We now formally state the lemma and then prove it.
\begin{lemma}\label{lem:SJRmain}
    \vcmember is Karp reducible to \lamSJRBExist.
\end{lemma}

\paragraph{The construction.}
Fix an instance $(G=(V,E),x)$ of \vcmember, where $|V|=m_v$ and $|E|=m_e$.
Without loss of generality, we may assume $m_v,m_e\ge 6$ and index $V=\{v_1,\dots,v_{m_v}\}$.
We construct an instance $(\mathbb{E},\lambda,s)$ of \lamSJRBExist as follows.
For each $i,q\in[m_v]$, let $\varphi_{i,q}$ be a 3-CNF formula expressing that the complement graph $\overline{G}$ has a clique of size $q$ containing $v_i$.
Such a formula is obtained in the standard way by using variables to encode the vertex placed in each of the $q$ positions, clauses enforcing that every position receives exactly one vertex, the chosen vertices are distinct, every chosen pair is adjacent in $\overline{G}$, and one chosen vertex is $v_i$.
By the standard Tseitin transformation, replace each formula by an equisatisfiable proper 3-CNF formula.
If any resulting formula has no clause, conjoin a satisfiable proper clause over three fresh variables.
Next, add unused fresh variables so that all formulas have the same number \(n_0>2\) of variables.
Finally, since every formula now has at least one proper clause, duplicate existing proper clauses so that all formulas have the same number \(m_0\) of clauses and \(n_0\mid m_0\).
This preprocessing preserves satisfiability and keeps every clause proper.
Set $\gamma=7m_0+1$ and
\[
    \lambda=\frac{m_0}{n_0}+(n_0-1)\gamma+1.
\]
Let $\xi=n_0+m_v$.
The election instance $\mathbb{E}=(N,C,\mathcal{A},k)$ is described below.

\begin{enumerate}
    \item \textbf{Voters.}
    For a symbol $y$ and a positive integer $q$, let \(\langle y\rangle_q\) denote a block of \(q\) identical voters; when \(q=\lambda\), we simply write \(\langle y\rangle\).
    For each $i\in[m_v]$, define the following pairwise disjoint voter blocks:
    \[
        N^i \;=\; N_G^i \,\cup\, N_f^i \,\cup\, N_x^i
        \,\cup\, N_D^i \,\cup\, N_z^i
        \,\cup\, N_A^i \,\cup\, N_L^i \,\cup\, N_T^i,
    \]
    where
    \[
    \begin{aligned}
        N_G^i &:= \{\langle g_e^i\rangle_{2\lambda} : e\in E\},\\
        N_f^i &:= \{\langle h_j^i\rangle : j\in[m_v]\},\\
        N_x^i &:= \{\langle x^i\rangle\},\\
        N_D^i &:= \{\langle\delta_j^i\rangle : j\in[m_v]\},\\
        N_z^i &:= \{\langle\eta^i\rangle\},\\
        N_A^i &:= \{\langle a_q^i\rangle : q\in[m_v]\},\\
        N_L^i &:= \{\langle\ell_q^i\rangle_{\lambda(q-1)} : q\in\{2,\ldots,m_v\}\},\\
        N_T^i &:= \bigcup_{q=1}^{m_v}T^{i,q}.
    \end{aligned}
    \]
    Here $T^{i,q}$ is a fresh copy of the verifier voters $T(\varphi_{i,q})$ from Lemma~\ref{lem:SJR-formula-gadget}.
    Finally, let $N_Y:=\{\langle\omega_t\rangle:t\in[\xi]\}$ and set $N:=N_Y\cup\bigcup_{i=1}^{m_v} N^i$.
    Since \(N\) contains the block \(\langle x^1\rangle\) of \(\lambda\) voters, the constructed instance satisfies \(\lambda\le |N|\).

    \item \textbf{Candidates.}
    For each $i\in[m_v]$, let
    \[
        C^i \;=\; \mathcal C^i \,\cup\, D^i \,\cup\, Z^i,
    \]
    where
    \[
        \mathcal C^i := C_G^i \cup C_{\mathrm{wit}}^i,
    \]
    \[
        C_G^i := \{\mathfrak{c}_v^i : v\in V\}
        \qquad\text{and}\qquad
        C_{\mathrm{wit}}^i := \{w_q^i : q\in[m_v]\},
    \]
    and
    \[
        D^i := \{d_j^i : j\in[m_v]\}
        \qquad\text{and}\qquad
        Z^i := \{z^i\}.
    \]
    For each $i,q\in[m_v]$, let \(C_{\mathrm{asgn}}^{i,q}\) be the assignment-candidate set \(C^{\mathrm{asgn}}(\varphi_{i,q})\) from Lemma~\ref{lem:SJR-formula-gadget}, and let
    \[
        F^{i,q}:=\{b_r^{i,q}:r\in[q-1]\}
    \]
    be a set of auxiliary filler candidates, with $F^{i,1}=\emptyset$.
    Let $Y:=\{y_t:t\in[\xi]\}$ be a set of forced padding candidates.
    Set
    \[
        C:=Y\cup \bigcup_{i=1}^{m_v} C^i
        \cup \bigcup_{i,q\in[m_v]}(C_{\mathrm{asgn}}^{i,q}\cup F^{i,q}).
    \]

    The roles of these voter and candidate families are as follows.
    The edge blocks \(N_G^i\), together with the vertex candidates \(C_G^i\), enforce a vertex cover of \(G\) inside copy \(i\), and \(N_x^i\) forces this cover to contain the distinguished vertex \(x\).
    The budget blocks \(N_f^i\) force exactly \(m_v\) selected candidates from the active set \(\mathcal C^i=C_G^i\cup C_{\mathrm{wit}}^i\).
    For each level \(q\), the verifier voters \(T^{i,q}\), the blocks in \(N_A^i\) and \(N_L^i\), and the candidates \(C_{\mathrm{asgn}}^{i,q}\cup F^{i,q}\cup\{w_q^i\}\) implement the clique-certificate module that can force the clique-witness candidate \(w_q^i\).
    The blocks \(N_D^i\), \(N_z^i\), and \(N_Y\) force padding candidates, leaving only the intended active seats.

    \item \textbf{Approval sets.}
    For each $i\in[m_v]$, define approvals for each voter block as follows:
    \begin{itemize}
        \item For every edge $e=(u,v)\in E$,
        \[
            A_{\langle g_e^i\rangle_{2\lambda}} \;=\; \{z^i,\mathfrak{c}_u^i,\mathfrak{c}_v^i\}.
        \]
        \item For every $j\in[m_v]$,
        \[
            A_{\langle h_j^i\rangle} \;=\; \mathcal C^i.
        \]
        \item For the distinguished block,
        \[
            A_{\langle x^i\rangle} \;=\; \{\mathfrak{c}_x^i\}.
        \]
        \item For the forced dummy blocks,
        \[
            A_{\langle\delta_j^i\rangle}=\{d_j^i\}\quad\text{for every }j\in[m_v],
            \qquad
            A_{\langle\eta^i\rangle}=\{z^i\}.
        \]
        \item For every certificate level $q\in[m_v]$,
        \[
            A_{\langle a_q^i\rangle}=C_{\mathrm{asgn}}^{i,q}\cup F^{i,q}\cup\{w_q^i\}.
        \]
        For $q\ge2$, the padding block has the same approval set:
        \[
            A_{\langle\ell_q^i\rangle_{\lambda(q-1)}}=C_{\mathrm{asgn}}^{i,q}\cup F^{i,q}\cup\{w_q^i\}.
        \]
        Finally, for each verifier voter $t\in T^{i,q}$, if \(A_t^{\mathrm{asgn}}\subseteq C_{\mathrm{asgn}}^{i,q}\) is its approval set in the formula gadget of Lemma~\ref{lem:SJR-formula-gadget}, set
        \[
            A_t=A_t^{\mathrm{asgn}}\cup F^{i,q}\cup\{w_q^i\}.
        \]
    \end{itemize}
    In addition, for every forced padding candidate $y_t\in Y$, set
    \[
        A_{\langle\omega_t\rangle}=\{y_t\}.
    \]
    The family of approval sets $\mathcal{A}$ is given by collecting the above sets over all voter blocks.

    \item \textbf{Committee size.}
    Set
    \[
        k=m_v^2+m_v(m_v+1)+\xi=2m_v^2+m_v+\xi.
    \]
\end{enumerate}

\noindent
Lastly, define the base satisfaction function $s$ by
\[\label{eqn:SJR-base-satisfaction}
        s(v)=\begin{cases}
        n_0+q-1, & \text{if } v\in \langle a_q^i\rangle \text{ for some } i,q\in[m_v],\\
        k, & \text{if } v\in T^{i,q} \text{ for some } i,q\in[m_v],\\
        k, & \text{if } v\in \langle\ell_q^i\rangle_{\lambda(q-1)} \text{ for some } i\in[m_v],\,q\ge2,\\
        0, & \text{otherwise.}
        \end{cases}
\]
The total value \(\sum_{v\in N}s(v)\) is polynomially bounded in the constructed instance size, since all voters are explicitly present and every base value is at most \(k\).

Our construction is illustrated schematically in Fig.~\ref{fig:SJRconstruction} by using an example with $|V|=3$.

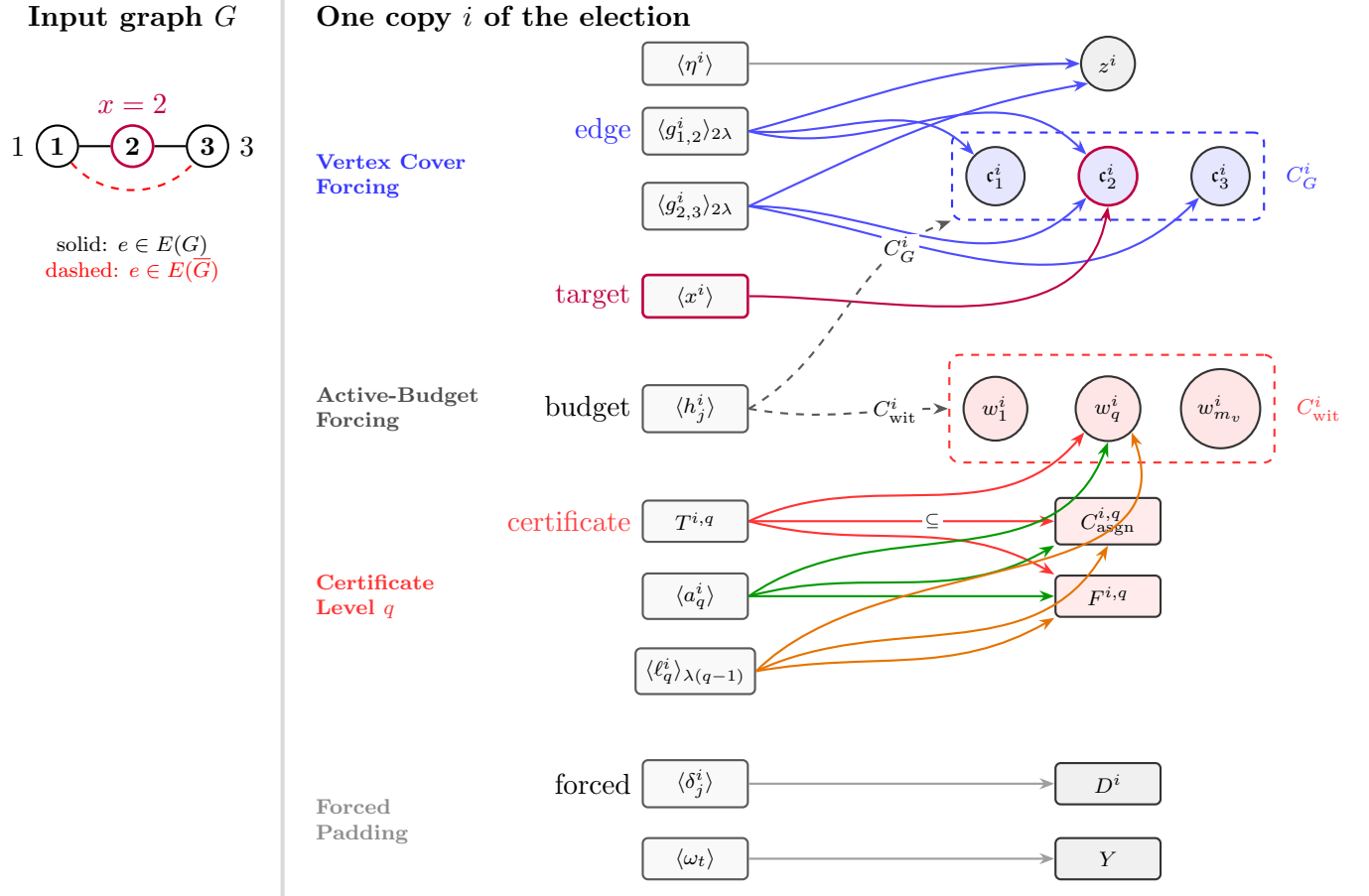
\begin{figure}[ht!]
\begin{tikzpicture}[
        thick,
        >={Stealth[length=2mm]},
        vertex/.style={circle, draw=black, fill=white, inner sep=0pt, minimum size=0.55cm, font=\small\bfseries},
        voter/.style={rectangle, draw=black!65, fill=gray!5, rounded corners=2pt, minimum width=1.4cm, minimum height=0.55cm, font=\scriptsize, align=center},
        candidate/.style={circle, draw=black!80, fill=white, minimum size=0.62cm, font=\scriptsize, align=center},
        candset/.style={rectangle, draw=black!80, fill=white, rounded corners=2pt, minimum width=1.4cm, minimum height=0.55cm, font=\scriptsize, align=center},
        edge_app/.style={draw=blue!70, thick, ->},
        target_app/.style={draw=purple, thick, ->},
        budget_app/.style={draw=black!65, thick, dashed, ->},
        forced_app/.style={draw=gray!75, thick, ->},
        cert_t/.style={draw=red!80, thick, ->},
        cert_a/.style={draw=green!60!black, thick, ->},
        cert_ell/.style={draw=orange!90!black, thick, ->}
]

% ==========================================
% 1. Input Graph G
% ==========================================
\node[font=\bfseries] at (-2.0, 7.2) {Input graph $G$};
\node[vertex, label=left:1] (v1) at (-3.0, 5.5) {1}; 
\node[vertex, draw=purple, line width=1pt, label=above:\textcolor{purple}{$x=2$}] (v2) at (-2.0, 5.5) {2};
\node[vertex, label=right:3] (v3) at (-1.0, 5.5) {3};
\draw[thick] (v1) -- (v2) -- (v3);
\draw[dashed, red] (v1) to[bend right=50] (v3); 
\node[align=center, font=\scriptsize] at (-2.0, 4.0) {solid: $e \in E(G)$\\ \textcolor{red}{dashed: $e \in E(\overline{G})$}};

% ==========================================
% 2. Divider & Structure Layout
% ==========================================
\draw[gray!30, line width=1.5pt] (0.0, 7.5) -- (0.0, -4.5);
\node[font=\bfseries, anchor=west] at (0.3, 7.2) {One copy $i$ of the election};

% Section Labels (Anchored west, fixed at x=0.3)
\node[font=\scriptsize\bfseries, blue!80, align=left, anchor=west] at (0.3, 5.1) {Vertex Cover\\Forcing};
\node[font=\scriptsize\bfseries, black!70, align=left, anchor=west] at (0.3, 2.0) {Active-Budget\\Forcing};
\node[font=\scriptsize\bfseries, red!80, align=left, anchor=west] at (0.3, -0.5) {Certificate\\Level $q$};
\node[font=\scriptsize\bfseries, gray!90, align=left, anchor=west] at (0.3, -3.5) {Forced\\Padding};

% ==========================================
% 3. Voters (Shifted right to x=5.5 to completely clear the text)
% ==========================================
\node[voter] (veta) at (5.5, 6.6) {$\langle\eta^i\rangle$};
\node[voter, label=left:\textcolor{blue!80}{edge}] (vg12) at (5.5, 5.7) {$\langle g_{1,2}^i\rangle_{2\lambda}$};
\node[voter] (vg23) at (5.5, 4.7) {$\langle g_{2,3}^i\rangle_{2\lambda}$};
\node[voter, draw=purple, line width=1pt, label=left:\textcolor{purple}{target}] (vx) at (5.5, 3.5) {$\langle x^i\rangle$};
\node[voter, label=left:budget] (vh) at (5.5, 2.0) {$\langle h_j^i\rangle$};
\node[voter, label=left:\textcolor{red!75}{certificate}] (vt) at (5.5, 0.5) {$T^{i,q}$};
\node[voter] (va) at (5.5, -0.5) {$\langle a_q^i\rangle$};
\node[voter] (vell) at (5.5, -1.5) {$\langle\ell_q^i\rangle_{\lambda(q-1)}$};
\node[voter, label=left:forced] (vdelta) at (5.5, -3.0) {$\langle\delta_j^i\rangle$};
\node[voter] (vomega) at (5.5, -4.0) {$\langle\omega_t\rangle$};

% ==========================================
% 4. Candidates (Shifted appropriately to balance the diagram)
% ==========================================
\node[candidate, fill=gray!12] (zcand) at (11.0, 6.6) {$z^i$};

% C_G^i Candidates & Box
\node[candidate, fill=blue!10] (cg1) at (9.5, 5.1) {$\mathfrak c_1^i$};
\node[candidate, fill=blue!10, draw=purple, line width=1pt] (cg2) at (11.0, 5.1) {$\mathfrak c_2^i$};
\node[candidate, fill=blue!10] (cg3) at (12.5, 5.1) {$\mathfrak c_3^i$};
\node[draw=blue!80, dashed, rounded corners, fit=(cg1) (cg2) (cg3), inner sep=5pt] (cgbox) {};
% PLACED ON THE RIGHT
\node[font=\scriptsize\bfseries, blue!80, anchor=west] at ([xshift=4pt]cgbox.east) {$C_G^i$};

% C_{\mathrm{wit}}^i Candidates & Box
\node[candidate, fill=red!10] (p1) at (9.5, 2.0) {$w_1^i$};
\node[candidate, fill=red!10] (pq) at (11.0, 2.0) {$w_q^i$};
\node[candidate, fill=red!10] (pm) at (12.5, 2.0) {$w_{m_v}^i$};
\node[draw=red!80, dashed, rounded corners, fit=(p1) (pq) (pm), inner sep=5pt] (pbox) {};
% PLACED ON THE RIGHT
\node[font=\scriptsize\bfseries, red!80, anchor=west] at ([xshift=4pt]pbox.east) {$C_{\mathrm{wit}}^i$};

% Certificate and filler candidate sets
\node[candset, fill=red!8] (bset) at (11.0, 0.5) {$C_{\mathrm{asgn}}^{i,q}$};
\node[candset, fill=red!8] (fset) at (11.0, -0.5) {$F^{i,q}$};
\node[candset, fill=gray!12] (dcand) at (11.0, -3.0) {$D^i$};
\node[candset, fill=gray!12] (ycand) at (11.0, -4.0) {$Y$};

% ==========================================
% 5. Routed Approval Arrows (Collision-Free)
% ==========================================
% --- Forced top ---
\draw[forced_app] (veta) -- (zcand);

% --- Blue Edge Routes ---
\draw[edge_app] (vg12.east) to[out=15, in=180] (zcand.west);
\draw[edge_app] (vg12.east) to[out=-5, in=135] (cg1.north west);
\draw[edge_app] (vg12.east) to[out=-15, in=135] (cg2.north west);

\draw[edge_app] (vg23.east) to[out=25, in=190] (zcand.south west);
\draw[edge_app] (vg23.east) to[out=0, in=225] (cg2.south west);
\draw[edge_app] (vg23.east) to[out=-10, in=225] (cg3.south west);

% --- Target x ---
\draw[target_app] (vx.east) to[out=0, in=260] (cg2.south);

% --- Budget h ---
% Dashed line to C_{\mathrm{wit}}^i
\draw[budget_app] (vh.east) to[out=-10, in=180] node[pos=0.75, fill=white, inner sep=1pt, font=\scriptsize] {$C_{\mathrm{wit}}^i$} (pbox.west);
% Dashed line to C_G^i (Now exiting from the right and sweeping upwards to avoid the target node!)
\draw[budget_app] (vh.east) to[out=30, in=200] node[pos=0.8, fill=white, inner sep=1pt, font=\scriptsize] {$C_G^i$} (cgbox.south west);

% --- Color-Coded Certificate Routes ---
\draw[cert_t] (vt.east) to[out=25, in=230] (pq.south west);
\draw[cert_t] (vt.east) -- node[pos=0.6, fill=white, inner sep=1pt, font=\tiny, text=black] {$\subseteq$} (bset.west);
\draw[cert_t] (vt.east) to[out=-15, in=150] (fset.north west);

\draw[cert_a] (va.east) to[out=35, in=250] (pq.south);
\draw[cert_a] (va.east) to[out=15, in=210] (bset.south west);
\draw[cert_a] (va.east) -- (fset.west);

\draw[cert_ell] (vell.east) to[out=45, in=290] (pq.south east);
\draw[cert_ell] (vell.east) to[out=25, in=240] (bset.south);
\draw[cert_ell] (vell.east) to[out=10, in=210] (fset.south west);

% --- Forced bottom ---
\draw[forced_app] (vdelta) -- (dcand);
\draw[forced_app] (vomega) -- (ycand);

% Footer Notes
\node[anchor=west, font=\footnotesize] at (0.3, -5) {The copy is duplicated for every $i\in[m_v]$; certificate levels are duplicated for every $q\in[m_v]$.};
% \node[anchor=west, font=\scriptsize, text=gray] at (0.3, -5.4) {Note: Bottom three certificate modules are color-coded (red, dark green, orange) to clarify their respective routing.};

\end{tikzpicture}
\caption{Schematic construction for Lemma~\ref{lem:SJRmain}, showing one representative copy and one representative certificate level with $m_v=3$.}
\label{fig:SJRconstruction}
\end{figure}

\paragraph{High-level ideas.}
We now explain the main intuition behind the construction.
Each copy \(i\) contains two types of active candidates,
\[
\mathcal C^i=C_G^i\cup C_{\mathrm{wit}}^i.
\]
The candidates in \(C_G^i\) represent vertices of \(G\), while the candidates in \(C_{\mathrm{wit}}^i\) are clique-witness candidates forced by clique certificates in \(\overline G\).
The forced candidates \(Y\), \(D^i\), and \(Z^i\) consume all non-active committee positions.
The budget voters \(\langle h_j^i\rangle\) approve all candidates in \(\mathcal C^i\), and together form a \((\lambda,m_v)\)-group.
Consequently, every feasible committee selects exactly \(m_v\) active candidates from each copy.

The vertex-cover side uses the edge voters and the distinguished block.
The block \(\langle\eta^i\rangle\) forces \(z^i\).
For each edge \(e=(u,v)\in E\), the block \(\langle g_e^i\rangle_{2\lambda}\) approves
\[
\{z^i,\mathfrak c_u^i,\mathfrak c_v^i\}.
\]
Since this block is a \((\lambda,2)\)-group, once \(z^i\) is selected, SJR requires one of \(\mathfrak c_u^i\) and \(\mathfrak c_v^i\) to be selected as well.
The block \(\langle x^i\rangle\) forces \(\mathfrak c_x^i\).
Thus, for every feasible committee \(W\), the set
\[
S_i=\{v\in V:\mathfrak c_v^i\in W\}
\]
is a vertex cover of \(G\) containing \(x\).
If \(\pi\) denotes the minimum size of such a cover, then \(|W\cap C_G^i|\ge \pi\) in every copy.

The clique-certificate module uses the formula-verifier gadget.
For every copy \(i\) and level \(q\), the formula \(\varphi_{i,q}\) is satisfiable exactly when \(\overline G\) has a clique of size \(q\) containing \(v_i\).
If it is satisfiable, Lemma~\ref{lem:SJR-formula-gadget} gives a large verifier group inside \(T^{i,q}\) with \(n_0\) common assignment candidates in \(C_{\mathrm{asgn}}^{i,q}\).
Together with the block \(\langle a_q^i\rangle\) and the padding block \(\langle\ell_q^i\rangle_{\lambda(q-1)}\), this creates a \((\lambda,n_0+q)\)-group whose common approved candidates are the \(n_0\) assignment candidates, the \(q-1\) filler candidates in \(F^{i,q}\), and the clique-witness candidate \(w_q^i\).
No candidate in \(C_{\mathrm{asgn}}^{i,q}\cup F^{i,q}\) can be selected by the committee-size accounting, so the only selected candidate that can raise voters in \(\langle a_q^i\rangle\) from base satisfaction \(n_0+q-1\) to \(n_0+q\) is \(w_q^i\).
Therefore every feasible committee must select \(w_1^i,\ldots,w_{\tau_i}^i\), where
\[
\tau_i=\max\{|K|:K\subseteq V,\ K\text{ is a clique in }\overline G,\ v_i\in K\}.
\]

The soundness contradiction comes from making the vertex-cover requirement and the clique-certificate requirement compete for the same \(m_v\) active seats.
Choose \(k^*\) such that \(v_{k^*}\) belongs to a maximum clique of \(\overline G\).
Then \(\tau_{k^*}=\alpha(G)\), where \(\alpha(G)\) is the independence number of \(G\).
Since exactly \(m_v\) active candidates are selected in copy \(k^*\), and at least \(\alpha(G)\) of them are clique-witness candidates, at most \(m_v-\alpha(G)=\beta(G)\) vertex candidates can be selected there.
If no minimum vertex cover of \(G\) contains \(x\), then \(\pi>\beta(G)\), contradicting the vertex-cover forcing requirement \(|W\cap C_G^{k^*}|\ge\pi\).

Conversely, suppose that \(G\) has a minimum vertex cover \(\mathfrak C\) containing \(x\).
Let \(I=V\setminus\mathfrak C\), so \(|I|=\alpha(G)\).
In copy \(i\), the intended committee selects all forced candidates in \(D^i\cup Z^i\), the vertex candidates corresponding to \(\mathfrak C\), the clique-witness candidates \(w_1^i,\ldots,w_{\tau_i}^i\), and any \(\alpha(G)-\tau_i\) additional vertex candidates from \(I\).
The active part then has size
\[
|\mathfrak C|+(\alpha(G)-\tau_i)+\tau_i
=\beta(G)+\alpha(G)
=m_v.
\]
The selected cover candidates satisfy the edge and target groups, the budget voters receive all \(m_v\) active candidates, and every certificate group that can demand \(w_q^i\) has \(q\le\tau_i\), so the corresponding clique-witness candidate has already been selected.

\paragraph{A working example.}
Consider the three-vertex path in Fig.~\ref{fig:SJRconstruction}, with
\[
    V=\{v_1,v_2,v_3\},\qquad
    E=\{(v_1,v_2),(v_2,v_3)\},\qquad x=v_2.
\]
For this illustration, we apply the construction directly with $m_v=3$
and $m_e=2$, setting aside the size assumption $m_v,m_e\ge6$ used in the
general reduction.
The unique minimum vertex cover is $\mathfrak{C}=\{v_2\}$, so this is a
\emph{yes-instance} of \vcmember, with
\[
    \beta(G)=1,\qquad \alpha(G)=2.
\]
The complement graph has the single edge $(v_1,v_3)$. Thus, the largest
complement cliques containing $v_1,v_2,v_3$ have sizes
\[
    (\tau_1,\tau_2,\tau_3)=(2,1,2).
\]
We retain the normalized formula parameters $n_0,m_0$ and the quota
$\lambda$ from the construction. In particular,
\[
    \xi=n_0+3,\qquad k=21+\xi=n_0+24.
\]

First, the singleton-approval blocks force every candidate in
\[
    Y\cup\bigcup_{i=1}^3(D^i\cup Z^i)
\]
to be selected. These candidates occupy $\xi+12$ seats, leaving exactly
$9$ active seats in total.
For each copy $i$, the three budget blocks contain $3\lambda$ voters
who commonly approve all six candidates in
$\mathcal C^i=C_G^i\cup C_{\mathrm{wit}}^i$.
They form a $(\lambda,3)$-group with zero base satisfaction, forcing
at least three active candidates from that copy.
Consequently, every \lamSJRB committee selects exactly three active
candidates in each copy and no candidate from any
$C_{\mathrm{asgn}}^{i,q}\cup F^{i,q}$.
Moreover, the distinguished block $\langle x^i\rangle$ forces
$\mathfrak{c}_{v_2}^i$. Together with the already forced $z^i$, this
candidate gives every edge voter satisfaction at least $2$, as required
by the corresponding $(\lambda,2)$-group.

To see concretely how a clique certificate forces a witness candidate,
consider copy $i=1$ and level $q=2$.
The clique $\{v_1,v_3\}$ in $\overline G$ makes $\varphi_{1,2}$
satisfiable. Choose a satisfying assignment encoding this clique.
Lemma~\ref{lem:SJR-formula-gadget} gives its canonical verifier group
$T'\subseteq T^{1,2}$, with
\[
    |T'|=\lambda n_0,
    \qquad
    \left|\bigcap_{t\in T'}A_t^{\mathrm{asgn}}\right|=n_0.
\]
Define
\[
    H=T'\cup\langle a_2^1\rangle
         \cup\langle\ell_2^1\rangle_{\lambda}.
\]
Then $|H|=\lambda(n_0+2)$, and its common approved candidates are
\[
    \bigcap_{u\in H}A_u
    =\left(\bigcap_{t\in T'}A_t^{\mathrm{asgn}}\right)
      \cup\{b_1^{1,2},w_2^1\}.
\]
Thus, $H$ is a $(\lambda,n_0+2)$-group.
Every voter in $\langle a_2^1\rangle$ has base satisfaction $n_0+1$.
Since the committee-size accounting excludes all assignment and filler
candidates, the only candidate that can supply the missing unit of
satisfaction is $w_2^1$.
Hence, any \lamSJRB committee must select $w_2^1$; otherwise, every
$u\in\langle a_2^1\rangle$ would have
\[
    S_W(u)+s(u)=n_0+1<n_0+2,
\]
violating the requirement of $H$.
The same argument at the other satisfiable certificate levels forces
\[
    w_1^1,w_2^1\quad\text{in copy }1,\qquad
    w_1^2\quad\text{in copy }2,\qquad
    w_1^3,w_2^3\quad\text{in copy }3.
\]

We now construct a feasible committee. Include all forced candidates,
and choose the active candidates as follows:
\begin{center}
\begin{tabular}{cll}
\toprule
Copy $i$ & Selected vertex candidates & Selected clique-witness candidates \\
\midrule
$1$ & $\{\mathfrak{c}_{v_2}^1\}$
    & $\{w_1^1,w_2^1\}$ \\
$2$ & $\{\mathfrak{c}_{v_2}^2,\mathfrak{c}_{v_1}^2\}$
    & $\{w_1^2\}$ \\
$3$ & $\{\mathfrak{c}_{v_2}^3\}$
    & $\{w_1^3,w_2^3\}$ \\
\bottomrule
\end{tabular}
\end{center}
Let $W^i$ consist of $D^i\cup Z^i$ and the active candidates in row $i$,
and set
\[
    W=Y\cup W^1\cup W^2\cup W^3.
\]
Copy $2$ needs only one clique-witness candidate, so we use its remaining
active seat for $\mathfrak{c}_{v_1}^2$.
Every copy has four forced candidates and three active candidates;
hence $|W|=\xi+3\cdot7=k$.
For the specific group $H$ above, every $u\in\langle a_2^1\rangle$
now has
\[
    S_W(u)+s(u)=1+(n_0+1)=n_0+2,
\]
while the other members of $H$ already have base satisfaction
$k\ge n_0+2$.
Thus, every member of $H$ meets its individual satisfaction requirement.

The remaining requirements follow the same pattern, as formalized in
the pointwise verification below.
The forced and distinguished voters receive satisfaction $1$, the edge
voters receive at least $2$, and the budget voters receive exactly $3$.
Verifier voters and certificate-padding voters have base satisfaction
$k$, which covers every cohesiveness level they can attain.
For $q\le\tau_i$, the selected $w_q^i$ raises each voter in
$\langle a_q^i\rangle$ to satisfaction $n_0+q$.
For $q>\tau_i$, the formula $\varphi_{i,q}$ is unsatisfiable, so the
formula-verifier property prevents these voters from belonging to a
$(\lambda,n_0+q)$-group or a group of higher cohesiveness.
Their base satisfaction $n_0+q-1$ therefore suffices.
Consequently, $W$ satisfies \lamSJRB.

For contrast, keep the same graph but change the distinguished vertex
to $x=v_1$.
The unique minimum vertex cover $\{v_2\}$ no longer contains $x$, so
this is a \emph{no-instance} of \vcmember.
Only the approvals of the distinguished blocks change; the clique
certificates and the active-seat budgets remain the same.
In copy $1$, the distinguished block now forces
$\mathfrak{c}_{v_1}^1$, and the edge block for $(v_2,v_3)$ requires at
least one of $\mathfrak{c}_{v_2}^1,\mathfrak{c}_{v_3}^1$ in addition
to $z^1$.
Thus, at least two vertex candidates must be selected in that copy.
The singleton clique $\{v_1\}$ still forces $w_1^1$, and the clique
$\{v_1,v_3\}$ still forces $w_2^1$ via the group $H$ above.
Any \lamSJRB committee $W'$ would therefore satisfy
\[
    |W'\cap\mathcal C^1|
    =|W'\cap C_G^1|+|W'\cap C_{\mathrm{wit}}^1|
    \ge2+2=4>3,
\]
contradicting the three-seat active budget in copy $1$.
This illustrates how the vertex-cover requirement and the
clique-certificate requirement compete for the same active seats.

Finally, Lemmas~\ref{lem:lamSJRBtoLamSJR} and~\ref{lem:lamSJRtoSJR}
remove the base satisfactions and adjust the Hare quota, respectively,
while preserving committee existence.
Hence, the two choices of the distinguished vertex yield a yes-instance
and a no-instance of \SJRE, respectively.

\paragraph{Formal proof for the validity of the construction.}
Before proving the correctness of the reduction, we show the following propositions.

\begin{proposition}\label{prop:forced_dummy}
    Let $W \subseteq C$ be any committee of size~$k$.
    If $W$ satisfies \lamSJRB, then $Y\subseteq W$ and $D^i\cup Z^i\subseteq W$ for every $i\in[m_v]$.
\end{proposition}

\begin{proof}
    For each $i,j\in[m_v]$, the block \(\langle\delta_j^i\rangle\) is a $(\lambda,1)$-group whose only approved candidate is $d_j^i$ and whose voters have zero base satisfaction.
    Hence $d_j^i\in W$.
    The same argument applied to \(\langle\eta^i\rangle\) for each $i\in[m_v]$ gives $z^i\in W$.
    Finally, each block \(\langle\omega_t\rangle\) is a $(\lambda,1)$-group whose only approved candidate is $y_t$, so $y_t\in W$ for every $t\in[\xi]$.
\end{proof}

\begin{proposition}\label{prop:candidate_num}
    Let $W \subseteq C$ be a committee of size $k$.
    If $W$ satisfies \lamSJRB, then for each index $i$, $W$ contains exactly $m_v$ candidates from the active set \(\mathcal C^i=C_G^i\cup C_{\mathrm{wit}}^i\).
    Moreover, no candidate from any \(C_{\mathrm{asgn}}^{i,q}\cup F^{i,q}\) is selected.
\end{proposition}

\begin{proof}
    Consider the set of voters $N_f^i$.
    Since all voters in $N_f^i$ share the approval set \(\mathcal C^i\) and $|N_f^i|=\lambda m_v$, they constitute a $(\lambda,m_v)$-group.
    These voters have zero base satisfaction, so every \lamSJRB committee $W$ must satisfy
    \[
        |W\cap \mathcal C^i|\ge m_v.
    \]
    By Proposition~\ref{prop:forced_dummy}, every feasible committee contains all candidates in
    \[
        Y\cup\bigcup_i(D^i\cup Z^i),
    \]
    whose total number is $\xi+m_v(m_v+1)$.
    Since $k=m_v^2+m_v(m_v+1)+\xi$, only $m_v^2$ seats remain after these forced candidates.
    The lower bounds $|W\cap \mathcal C^i|\ge m_v$ over the $m_v$ copies therefore have to be tight, and no seat remains for any candidate outside the forced and active sets.
\end{proof}

\begin{proposition}\label{prop:max_clique}
    Let $W \subseteq C$ be any committee of size~$k$.
    Then if $W$ satisfies \lamSJRB, for each index~$i$, it contains at least~$\tau_i$ members from~$C_{\mathrm{wit}}^i$, where
    \[
        \tau_i \;=\; \max\bigl\{\,|K| : K \subseteq V, \;K\text{ is a clique in }\overline{G},\;v_i\in K\bigr\}
    \]
    is the size of the largest clique in the complement graph~$\overline{G}$ that contains~$v_i$.
\end{proposition}

\begin{proof}
    Fix an index $i$.
    We show that $w_q^i\in W$ for every $q\le \tau_i$.
    Fix such a $q$.
    Since $q\le \tau_i$, the formula $\varphi_{i,q}$ is satisfiable.
    By Lemma~\ref{lem:SJR-formula-gadget}, there is a voter set $T'\subseteq T^{i,q}$ with $|T'|=\lambda n_0$ and exactly $n_0$ commonly approved candidates in \(C_{\mathrm{asgn}}^{i,q}\).
    Define
    \[
        N' \;=\; T'\cup \langle a_q^i\rangle\cup \langle\ell_q^i\rangle_{\lambda(q-1)},
    \]
    where the last block is omitted when $q=1$.
    Then $|N'|=\lambda n_0+\lambda+\lambda(q-1)=\lambda(n_0+q)$.
    The voters in $N'$ commonly approve the $n_0$ assignment candidates identified by $T'$, all $q-1$ filler candidates in $F^{i,q}$, and the clique-witness candidate $w_q^i$.
    Hence $N'$ is a $(\lambda,n_0+q)$-group.

    The voters in \(\langle a_q^i\rangle\) have base satisfaction $n_0+q-1$.
    By Proposition~\ref{prop:candidate_num}, no candidate in \(C_{\mathrm{asgn}}^{i,q}\cup F^{i,q}\) is selected.
    Thus the only selected candidate that can increase the satisfaction of voters in \(\langle a_q^i\rangle\) to $n_0+q$ is $w_q^i$.
    Therefore $w_q^i\in W$ for every $q\le\tau_i$, and so $|W\cap C_{\mathrm{wit}}^i|\ge\tau_i$.
\end{proof}

\begin{proposition}\label{prop:min_cover}
    Let $W \subseteq C$ be any committee of size~$k$.
    Then if $W$ satisfies \lamSJRB, for each index~$i$, it contains at least~$\pi$ members from~$C_G^i$, where
    \[
        \pi \;=\; \min\bigl\{\,|K| : K \subseteq V, \;K\text{ is a vertex cover of }G,\;x\in K\bigr\}
    \]
    is the size of the smallest vertex cover in the graph~$G$ that contains~$x$.
\end{proposition}

\begin{proof}
    Fix an index $i \in [m_v]$ and define $S_i := \{v \in V : \mathfrak{c}_v^i \in W\}$.
    The block \(\langle x^i\rangle\) is a $(\lambda,1)$-group with approval set $\{\mathfrak{c}_x^i\}$ and zero base satisfaction, so $\mathfrak{c}_x^i\in W$ and $x\in S_i$.

    Next, let $e=(u,v)\in E$.
    The block \(\langle g_e^i\rangle_{2\lambda}\) is a $(\lambda,2)$-group: it has $2\lambda$ voters and common approval set $\{z^i,\mathfrak{c}_u^i,\mathfrak{c}_v^i\}$.
    Since $z^i\in W$ by Proposition~\ref{prop:forced_dummy}, \lamSJRB requires every voter in \(\langle g_e^i\rangle_{2\lambda}\) to approve at least one of $\mathfrak{c}_u^i$ and $\mathfrak{c}_v^i$ in addition to $z^i$.
    Thus $u\in S_i$ or $v\in S_i$.
    Since this holds for every edge of $G$, the set $S_i$ is a vertex cover containing $x$.
    Therefore $|W\cap C_G^i|=|S_i|\ge \pi$.
\end{proof}

\noindent Now, we are ready to show the soundness and completeness of the reduction.

\paragraph{Soundness.}
Assume, for the sake of contradiction, that no minimum vertex cover of $G$ contains the distinguished vertex $x$. Equivalently, if $\beta(G)$ denotes the size of a minimum vertex cover in $G$, then every minimum cover necessarily omits $x$. 

By Proposition~\ref{prop:max_clique}, any \lamSJRB committee $W$ must satisfy:
\[
        |W \cap C_{\mathrm{wit}}^i| \ge \tau_i \quad \text{for each } i = 1, \dots, m_v,
\]
where $\tau_i$ represents the maximum clique size in the complement graph $\overline{G}$ that includes vertex $v_i$. Specifically, let $k_*$ be an index such that $\tau_{k_*} = \alpha(G)$, where $\alpha(G)$ is the independence number of $G$ (equivalent to the clique number of $\overline{G}$). According to Proposition~\ref{prop:candidate_num}, the committee satisfies $|W \cap \mathcal C^{k_*}| = m_v$. It follows that:
\[
        |W \cap C_{\mathrm{wit}}^{k_*}| \ge \alpha(G) \implies |W \cap C_G^{k_*}| \le m_v - \alpha(G) = \beta(G).
\]

Conversely, Proposition~\ref{prop:min_cover} ensures that:
\[
        |W \cap C_G^{k_*}| \ge \pi,
\]
where $\pi$ denotes the size of a minimum vertex cover of $G$ constrained to contain $x$. By our initial assumption, $\pi > \beta(G)$. This leads to the following chain of inequalities:
\[
        \beta(G) < \pi \le |W \cap C_G^{k_*}| \le \beta(G),
\]
which constitutes a clear contradiction. Consequently, there exists no \lamSJRB committee of size $k$ under the given conditions.

\paragraph{Completeness.}
Conversely, suppose there is a minimum vertex cover $\mathfrak{C}$ of $G$ containing the vertex $x$.
Let
\[
    \mathfrak{V} = V \setminus \mathfrak{C} = \{t_1,\dots,t_d\},
\]
with $d=m_v-\beta(G)$.
For each~$i$, let
\[
        \tau_i=\max\bigl\{\,|K| : K \subseteq V,\; K\text{ is a clique in }\overline{G},\; v_i\in K\bigr\}.
\]
Choose an arbitrary set $F_i\subseteq \mathfrak{V}$ of size $d-\tau_i$, which is possible because $\tau_i\le d=\alpha(G)$.
We construct the committee
\[
    W = Y\cup \bigcup_{i=1}^{m_v} W^i,
\]
where
\[
    W^i =
        D^i\cup Z^i
        \cup \{\mathfrak{c}_v^i : v\in \mathfrak{C}\cup F_i\}
        \cup \{w_q^i : q\in[\tau_i]\}.
\]
The active part of $W^i$ has size
\[
    |\mathfrak{C}|+|F_i|+\tau_i
    =\beta(G)+(d-\tau_i)+\tau_i
    =m_v,
\]
and $|D^i\cup Z^i|=m_v+1$.
Therefore $|W|=\xi+m_v(2m_v+1)=2m_v^2+m_v+\xi=k$.
It remains to show $W$ satisfies every group requirement. 
Since candidate sets are disjoint across different copies and the candidates in $Y$ are approved only by the corresponding forced padding blocks, every \((\lambda,\ell)\)-group with $\ell\ge1$ is contained in one copy or in one forced padding block.
Thus, it suffices to check each voter type:
\begin{itemize}
    \item Each forced dummy block \(\langle\delta_j^i\rangle\), \(\langle\eta^i\rangle\), and \(\langle\omega_t\rangle\) approves a unique candidate, and that candidate is included in $W$.
    \item Distinguished block \(\langle x^i\rangle\) has its unique approved candidate included in~$W^i$.
    \item Edge-voter blocks have cohesiveness at most~2. If a group contains voters only from one edge block, then its size is at most $2\lambda$; if it contains voters from several edge blocks or from other types, the common approval set has size at most~2. Since $z^i\in W^i$ and $\mathfrak{C}$ is a vertex cover, every edge voter approves $z^i$ and at least one selected endpoint candidate.
    \item Budget blocks \(\langle h_j^i\rangle\) approve only \(\mathcal C^i\). Any group containing a budget voter has cohesiveness at most $m_v$: if it contains only budget voters, its size is at most $\lambda m_v$; if it also contains edge, distinguished, or certificate voters, the common approval set has size at most~2,~1, or~1, respectively. Since $|W^i\cap \mathcal C^i|=m_v$, all such groups are satisfied.
    \item Verifier voters in $T^{i,q}$ and padding voters in \(\langle\ell_q^i\rangle_{\lambda(q-1)}\) have base satisfaction $k$, so they satisfy every possible requirement.
    Indeed, any group containing such a voter and a voter outside the same certificate module has cohesiveness at most~1.
    For groups within the module, the counting argument in the next item yields $\ell\le n_0+q\le n_0+m_v\le k$.
    \item It remains to consider a group $N'$ containing a voter from \(\langle a_q^i\rangle\).
    If $N'$ also contains a budget voter, its common approval set is contained in $\{w_q^i\}$, so its cohesiveness is at most~1 and the base satisfaction $n_0+q-1$ is enough.
    If $N'$ contains any edge, distinguished, forced dummy, or different certificate-level voter, then its common approval set is empty.
    Hence the only nontrivial case is
    \[
        N'\subseteq \langle a_q^i\rangle\cup \langle\ell_q^i\rangle_{\lambda(q-1)}\cup T^{i,q},
    \]
    again omitting the $\ell$-block when $q=1$.
    Let $N'$ be a $(\lambda,\ell)$-group.
    If $\ell\le n_0+q-1$, the base satisfaction of \(\langle a_q^i\rangle\) already suffices.
    Suppose $\ell\ge n_0+q$, and put $r=\ell-q$.
    At most $\lambda q$ voters in $N'$ can come from \(\langle a_q^i\rangle\cup \langle\ell_q^i\rangle_{\lambda(q-1)}\), so $|N'\cap T^{i,q}|\ge \lambda r$.
    Moreover, since the candidates outside \(C_{\mathrm{asgn}}^{i,q}\) that are commonly approved in this certificate module are exactly $F^{i,q}\cup\{w_q^i\}$, the voters in $N'\cap T^{i,q}$ have at least $r$ common candidates in \(C_{\mathrm{asgn}}^{i,q}\).
    If $r>n_0$, write $r=n_0+\rho$ for some integer $\rho\ge1$.
    Then $N'\cap T^{i,q}$ has size at least $\lambda(n_0+\rho)$ and has at least $n_0+\rho$ common candidates in \(C_{\mathrm{asgn}}^{i,q}\), contradicting the third bullet of Lemma~\ref{lem:SJR-formula-gadget}.
    Therefore $r=n_0$.
    By the second bullet of Lemma~\ref{lem:SJR-formula-gadget}, $\varphi_{i,q}$ is satisfiable, so $q\le\tau_i$.
    Thus $w_q^i\in W$, and voters in \(\langle a_q^i\rangle\) obtain satisfaction
    \[
        (n_0+q-1)+1=n_0+q=\ell.
    \]
\end{itemize}
Hence, every cohesive group is represented, so $W$ is a \lamSJRB committee of size~$k$.

This completes the proof of correctness of the reduction and thus establishes $\ThetaTwo$-hardness of \lamSJRBExist.

\subsection{Reduction from \lamSJRBExist to \SJRE}
\label{sec:SJR_basetoSJR}
%Finally, to prove Theorem~\ref{thm:sjr}, we reduce \lamSJRBExist to \SJRE.
Finally, the following two lemmas conclude Theorem~\ref{thm:sjr}.
The main ideas are as follows: we can add dummy voters and candidates to adjust the size of cohesive groups such that $\lambda=n/k$; 
for voters with positive base satisfactions, we add some dummy ``must select'' candidates that are approved by each such voter, and we add some dummy voters that only approve these candidates to make sure they must be selected.
These two transformations are analogous to the AJR transformations in Sect.~\ref{sec:AJR_basetoAJR}, with individual satisfaction replacing average satisfaction.

\begin{lemma}\label{lem:lamSJRBtoLamSJR}
    \lamSJRBExist is Karp reducible to \lamSJRExist.
\end{lemma}

\begin{proof}
    Let $(E,\lambda,s)$ be an instance of \lamSJRBExist, where $E=(N,C,\mathcal{A},k)$. 
    This is the individual-satisfaction analogue of Lemma~\ref{lem:lamAJRBtoLamAJR}; we keep the details needed for the pointwise SJR condition.
    By the definition of \lamSJRBExist, $\lambda\le |N|$ and $\sum_{i\in N}s(i)$ is polynomially bounded in the input size; in particular, every value $s(i)$ is polynomially bounded.
    Hence the voters and candidates added below form a polynomial-size instance.
    We construct an instance $(\mathbb{E},\lambda')$ of \lamSJRExist, with $\mathbb{E}=(N',C',\mathcal{A}',k')$, as follows. 
    Set
    \[
    k' = k + \sum_{i\in N} s(i), \qquad \lambda' = \lambda.
    \]
    For each voter $i\in N$ with $s(i)>0$ and for each $j\in\{1,\ldots,s(i)\}$, create a set $N_i^j=\{n^{(i,j)}_1,\ldots,n^{(i,j)}_{\lambda}\}$ of $\lambda$ new voters. 
    Let
    \[
    N' = N \cup \bigcup_{\substack{i\in N,\,s(i)>0\\ j\in[1,s(i)]}} N_i^j.
    \]
    Since \(\lambda'=\lambda\le |N|\le |N'|\), the output satisfies the input restriction of \lamSJRExist.
    For each $i\in N$ with $s(i)>0$, introduce a candidate set $C_i=\{c^i_1,\ldots,c^i_{s(i)}\}$; if $s(i)=0$, let $C_i=\emptyset$. 
    Then set
    \[
    C' = C \cup \bigcup_{i\in N} C_i.
    \]

    Define the approval sets as follows.
    For each $i\in N$, let $A'_i = A_i \cup C_i$.
    For each $n^{(i,j)}_k\in N_i^j$, set $A'_{n^{(i,j)}_k}=\{c^i_j\}$. 

    For every positive integer $\ell$ and every $G\subseteq N$, $G$ is a $(\lambda,\ell)$-group in $E$ if and only if it is a $(\lambda',\ell)$-group in $\mathbb{E}$.
    Indeed, an old-only $(\lambda',\ell)$-group has size at least $\lambda>2$, so it contains at least two old voters; the added candidate sets $C_i$ are pairwise disjoint across old voters, and therefore they do not change the common approved candidate set of any old-only group.
    Any other $(\lambda',\ell)$-group in $\mathbb{E}$ must include some voters from $\bigcup_{i,j} N_i^j$, but since each such voter approves exactly one candidate, the common approved candidate set of these groups has size at most $1$. 
    Hence, they are all $(\lambda',1)$-groups. 

    Suppose $W$ is a committee of size $k$ that satisfies \lamSJRB for $(E,\lambda,s)$. 
    Define \(W' = W \cup \bigcup_{i\in N} C_i\).
    Then $W'$ is a committee of size $k'$ for $\mathbb{E}$. 
    Consider any $(\lambda',\ell)$-group $G'$ in $\mathbb{E}$. 
    If $G'$ contains voters not in $N$, then by construction we know that $\ell=1$ and each voter from $G'$ approves at least one candidate in $W'$.
    Therefore, the group is satisfied by $W'$.
    If $G'$ is contained entirely within $N$, then $G'$ is an old $(\lambda,\ell)$-group.
    Since $W$ satisfies \lamSJRB, for every voter $i\in G'$ we have $S_W(i)+s(i)\ge\ell$.
    The selected candidates in $C_i$ contribute exactly $s(i)$ to voter $i$, so $S_{W'}(i)\ge\ell$ for every voter $i\in G'$.
    Hence, $W'$ satisfies $\lambda'$-SJR for $(\mathbb{E},\lambda')$.

    Conversely, suppose $W'$ is a committee of size $k'$ that satisfies $\lambda'$-SJR for $(\mathbb{E},\lambda')$. 
    Each candidate $c^i_j\in C_i$ is approved by exactly $\lambda$ voters in $N_i^j$, so $W'$ must contain all of $\bigcup_{i\in N} C_i$. 
    Thus, $|W'\cap C|=k$. Let $W=W'\cap C$. 
    For any $(\lambda, \ell)$-group $G$ in $E$, the corresponding group in the augmented instance $\mathbb{E}$ is satisfied by $W'$.
    For each voter $i\in G$, all candidates in $C_i$ are selected and contribute exactly $s(i)$ to $S_{W'}(i)$.
    Therefore
    \[
            S_W(i)+s(i)=S_{W'}(i)\ge \ell.
    \]
    This confirms that the satisfaction threshold $\ell$ is attained for every voter in the group, thereby establishing that $W$ satisfies the \lamSJRB criterion.

    Hence, we have the equivalence between $(E,\lambda,s)$ and $(\mathbb{E},\lambda')$, completing the proof.
\end{proof}

\begin{lemma}\label{lem:lamSJRtoSJR}
    \lamSJRExist is Karp reducible to \SJRE.
\end{lemma}
\begin{proof}
        
    Let \((\mathbb{E},\lambda)\) be an instance of \lamSJRExist with \(\mathbb{E}=(N,C,\mathcal{A},k)\) and \(N=\{1,\dots,n\}\).
    This is the quota-adjustment analogue of Lemma~\ref{lem:lamAJRtoAJR}, with individual satisfaction replacing average satisfaction.
    When $\lambda k=|N|$, the claim is trivial, since in this case \lamSJRExist coincides exactly with SJR. 
    We therefore focus on the remaining two cases. 

    If $\lambda k>|N|$, then there exists an integer $\varepsilon$ such that $\lambda k=|N|+\varepsilon$.
    We construct an instance $\mathbb{E}'=(N',C',\mathcal{A}',k')$ of \SJRE\ as follows. 
    Set $N'=N\cup N_0$, where $N_0=\{n_1,\ldots,n_\varepsilon\}$ consists of $\varepsilon$ fresh voters who approve no candidates. 
    Let $C'=C$, let $A'_i=A_i$ for all $i\in N$, and let $A'_i=\varnothing$ for all $i\in N_0$. 
    Finally, set $k'=k$. 
    The Hare quota in $\mathbb{E}'$ is then $|N'|/k'=\lambda$, and every $(\lambda,\ell)$-group in $\mathbb{E}$ corresponds to a unique $\ell$-cohesive group in $\mathbb{E}'$. 
    Thus, $\mathbb{E}'$ admits an SJR-satisfying committee $W$ if and only if $W$ is a \lamSJR committee for $\mathbb{E}$.

    If $\lambda k<|N|$, define
    \(
    \xi = |N| - \lambda k + 1
    \)
    Clearly, $\xi$ is a positive integer.
    We construct an instance $\mathbb{E}'=(N',C',\mathcal{A}',k')$ of \SJRE\ as follows. 
    Let
    \[
    N' = N \;\cup\; \bigcup_{i=1}^{\xi} X^i \;\cup\; \{q\},
    \]
    where each $X^i=\{x^i_1,\ldots,x^i_{\lambda-1}\}$ contains $\lambda-1$ new voters. 
    Extend the candidate set by $\xi$ new candidates $C'=C\cup \{y_1,\ldots,y_\xi\}$. 
    For approvals, keep $A'_i=A_i$ for all $i\in N$, set $A'_i=\{y_j\}$ for all $i\in X^j$, and set $A'_{q}=\{y_j:j\in\{1,\ldots,\xi\}\}$. 
    Finally, set $k'=k+\xi$. 

    The Hare quota in this construction is
    \begin{align*}
        \tfrac{|N'|}{k'}&=\tfrac{|N|+\xi(\lambda-1)+1}{k+\xi}\\
        &=\tfrac{\xi+\lambda k -1+\xi(\lambda-1)+1}{k+\xi}\\
        &=\lambda.
    \end{align*}

    Let $W$ be an SJR committee for $\mathbb{E}'$.
    Suppose $W$ does not contain all $y_1,\ldots,y_\xi$, then there exists some $y_j\notin W$ and the smallest satisfaction of \((\lambda,1)\)-group $X^j\cup \{q\}$ is 0,
    which contradicts the assumption that $W$ satisfies SJR.
    Hence, $W$ must contain all $y_1,\ldots,y_\xi$ and
    \(
    W^* = W - \{y_1,\ldots,y_\xi\}
    \)
    is a committee of size $k$ in $\mathbb{E}$.
    Let $G\subseteq N$ be any $(\lambda,\ell)$-group in $\mathbb{E}$.
    Since $\ell\lambda\le |G|\le |N'|=\lambda k'$, we have $\ell\le k'$, and $G$ is an $\ell$-cohesive group in $\mathbb{E}'$.
    Old voters approve no new candidates, so $S_{W^*}(i)=S_W(i)\ge\ell$ for every $i\in G$.
    Thus, $W^*$ satisfies \lamSJR in $\mathbb{E}$.

    Conversely, if $\mathbb{E}$ admits a \lamSJR committee $W^*$, then we can construct an SJR committee for $\mathbb{E}'$ by setting
    \(
    W = W^* \cup \{y_1,\ldots,y_\xi\}.
    \)
    Every old $\ell$-cohesive group in $\mathbb{E}'$ is a $(\lambda,\ell)$-group in $\mathbb{E}$.
    Since $W^*$ satisfies \lamSJR and old voters approve no new candidates, every member $i$ of such a group satisfies
    \[
        \ell\le S_{W^*}(i)=S_W(i)\le k.
    \]
    Thus, every old cohesive group has $\ell\le k$ and satisfies SJR under $W$.
    It remains to consider cohesive groups containing new voters.
    Old voters approve no candidates in $\{y_1,\ldots,y_\xi\}$, while new voters approve only candidates in this set, so no cohesive group contains both old and new voters.
    Among the new voters, any cohesive group containing a voter from $X^j$ or containing $q$ with common candidate $y_j$ is contained in $X^j\cup\{q\}$.
    Its common approval set is contained in $\{y_j\}$, so it can only be $1$-cohesive.
    Since $y_j\in W$, every such group is satisfied.
    Thus, $W$ satisfies SJR in $\mathbb{E}'$.

    In both cases, we obtain a reduction from \lamSJRExist to \SJRE, which establishes correctness.
\end{proof}

\bibliographystyle{alphaurl}
\bibliography{reference}

\newcommand{\etalchar}[1]{$^{#1}$}
\begin{thebibliography}{MSWW22}

\bibitem[ABC{\etalchar{+}}17]{aziz2017jr}
Haris Aziz, Markus Brill, Vincent Conitzer, Edith Elkind, Rupert Freeman, and Toby Walsh.
\newblock Justified representation in approval-based committee voting.
\newblock {\em Social Choice and Welfare}, 48(2):461--485, 2017.

\bibitem[AEH{\etalchar{+}}18]{aziz2018complexity}
Haris Aziz, Edith Elkind, Shenwei Huang, Martin Lackner, Luis~S{\'{a}}nchez Fern{\'{a}}ndez, and Piotr Skowron.
\newblock On the complexity of extended and proportional justified representation.
\newblock In Sheila~A. McIlraith and Kilian~Q. Weinberger, editors, {\em Proceedings of the Thirty-Second {AAAI} Conference on Artificial Intelligence, (AAAI-18), the 30th innovative Applications of Artificial Intelligence (IAAI-18), and the 8th {AAAI} Symposium on Educational Advances in Artificial Intelligence (EAAI-18), New Orleans, Louisiana, USA, February 2-7, 2018}, pages 902--909. {AAAI} Press, 2018.
\newblock URL: \url{https://doi.org/10.1609/aaai.v32i1.11478}, \href {https://doi.org/10.1609/AAAI.V32I1.11478} {\path{doi:10.1609/AAAI.V32I1.11478}}.

\bibitem[BFJL24]{brill2024phragmen}
Markus Brill, Rupert Freeman, Svante Janson, and Martin Lackner.
\newblock Phragm{\'e}n’s voting methods and justified representation.
\newblock {\em Mathematical programming}, 203(1):47--76, 2024.

\bibitem[BFKN19]{bredereck2019experimental}
Robert Bredereck, Piotr Faliszewski, Andrzej Kaczmarczyk, and Rolf Niedermeier.
\newblock An experimental view on committees providing justified representation.
\newblock In {\em IJCAI}, pages 109--115, 2019.

\bibitem[BGP26]{becker2026existence}
Patrick Becker, Matthias Greger, and Dominik Peters.
\newblock Existence of the core in approval-based committee elections.
\newblock {\em arXiv preprint arXiv:2609.11912}, 2026.

\bibitem[BH91]{buss1991truth}
Samuel~R Buss and Louise Hay.
\newblock On truth-table reducibility to {SAT}.
\newblock {\em Information and Computation}, 91(1):86--102, 1991.

\bibitem[BIMP25]{brill2025individual}
Markus Brill, Jonas Israel, Evi Micha, and Jannik Peters.
\newblock Individual representation in approval-based committee voting.
\newblock {\em Social Choice and Welfare}, 64(1):69--96, 2025.

\bibitem[BP23]{brill2023ejr+}
Markus Brill and Jannik Peters.
\newblock Robust and verifiable proportionality axioms for multiwinner voting.
\newblock In {\em Proceedings of the 24th {ACM} Conference on Economics and Computation, {EC}}, page 301. {ACM}, 2023.

\bibitem[Dro81]{droop1881methods}
Henry~Richmond Droop.
\newblock On methods of electing representatives.
\newblock {\em Journal of the Statistical Society of London}, 44(2):141--202, 1881.

\bibitem[EFI{\etalchar{+}}23]{elkind2023justifying}
Edith Elkind, Piotr Faliszewski, Ayumi Igarashi, Pasin Manurangsi, Ulrike Schmidt-Kraepelin, and Warut Suksompong.
\newblock Justifying groups in multiwinner approval voting.
\newblock {\em Theoretical Computer Science}, 969:114039, 2023.

\bibitem[FEL{\etalchar{+}}17]{fernandez2017pjr}
Luis~S{\'{a}}nchez Fern{\'{a}}ndez, Edith Elkind, Martin Lackner, Norberto~Fern{\'{a}}ndez Garc{\'{\i}}a, Jes{\'{u}}s Arias{-}Fisteus, Pablo Basanta{-}Val, and Piotr Skowron.
\newblock Proportional justified representation.
\newblock In {\em Proceedings of the Thirty-First {AAAI} Conference on Artificial Intelligence}, pages 670--676. {AAAI} Press, 2017.

\bibitem[Gil59]{gillies1959solutions}
Donald~B Gillies.
\newblock Solutions to general non-zero-sum games.
\newblock {\em Contributions to the Theory of Games}, 4(40):47--85, 1959.

\bibitem[Hem87]{hemachandra1987strong}
Lane~A Hemachandra.
\newblock The strong exponential hierarchy collapses.
\newblock In {\em Proceedings of the nineteenth annual ACM symposium on Theory of computing}, pages 110--122, 1987.

\bibitem[HSV05]{hemaspaandra2005complexity}
Edith Hemaspaandra, Holger Spakowski, and J{\"o}rg Vogel.
\newblock The complexity of kemeny elections.
\newblock {\em Theoretical Computer Science}, 349(3):382--391, 2005.

\bibitem[HTX{\etalchar{+}}26]{han2026likelihood}
Qishen Han, Biaoshuai Tao, Lirong Xia, Chengkai Zhang, and Houyu Zhou.
\newblock Likelihood of the existence of average justified representation.
\newblock In {\em Proceedings of the 2026 Annual ACM-SIAM Symposium on Discrete Algorithms (SODA)}, pages 1761--1794. SIAM, 2026.

\bibitem[JMW20]{JiangMW20}
Zhihao Jiang, Kamesh Munagala, and Kangning Wang.
\newblock Approximately stable committee selection.
\newblock In {\em Proceedings of the 52nd Annual {ACM} {SIGACT} Symposium on Theory of Computing, {STOC}}, pages 463--472. {ACM}, 2020.

\bibitem[Joh81]{johnson1981np}
David~S Johnson.
\newblock The {NP}-completeness column: an ongoing guide.
\newblock {\em Journal of Algorithms}, 2(4):393--405, 1981.

\bibitem[KLK25]{KalayciLK2025}
Yusuf~Hakan Kalayci, Jiasen Liu, and David Kempe.
\newblock Full proportional justified representation.
\newblock In {\em Proceedings of the 24th International Conference on Autonomous Agents and Multiagent Systems, {AAMAS}}, page 1070–1078, 2025.

\bibitem[LS23]{lackner2023multiwinner}
Martin Lackner and Piotr Skowron.
\newblock {\em Multi-Winner Voting with Approval Preferences}.
\newblock Springer Briefs in Intelligent Systems. Springer, 2023.

\bibitem[MMS23]{MavrovMS23}
Ivan{-}Aleksandar Mavrov, Kamesh Munagala, and Yiheng Shen.
\newblock Fair multiwinner elections with allocation constraints.
\newblock In {\em Proceedings of the 24th {ACM} Conference on Economics and Computation, {EC}}, pages 964--990. {ACM}, 2023.

\bibitem[MSWW22]{MunagalaSWW22}
Kamesh Munagala, Yiheng Shen, Kangning Wang, and Zhiyi Wang.
\newblock Approximate core for committee selection via multilinear extension and market clearing.
\newblock In {\em Proceedings of the 2022 {ACM-SIAM} Symposium on Discrete Algorithms, {SODA}}, pages 2229--2252. {SIAM}, 2022.

\bibitem[PPS20]{peters2020proportional}
Dominik Peters, Grzegorz Pierczynski, and Piotr Skowron.
\newblock Proportional participatory budgeting with cardinal utilities.
\newblock {\em arXiv preprint arXiv:2008.13276}, pages 2181--2188, 2020.

\bibitem[PPSS21]{PetersP0021}
Dominik Peters, Grzegorz Pierczynski, Nisarg Shah, and Piotr Skowron.
\newblock Market-based explanations of collective decisions.
\newblock In {\em Thirty-Fifth {AAAI} Conference on Artificial Intelligence}, pages 5656--5663. {AAAI} Press, 2021.

\bibitem[Puk17]{pukelsheim2017quota}
Friedrich Pukelsheim.
\newblock Quota methods of apportionment: Divide and rank.
\newblock In {\em Proportional Representation: Apportionment Methods and Their Applications}, pages 95--105. Springer, 2017.

\bibitem[SFJ{\etalchar{+}}22]{SzufaFJLSST22}
Stanislaw Szufa, Piotr Faliszewski, Lukasz Janeczko, Martin Lackner, Arkadii Slinko, Krzysztof Sornat, and Nimrod Talmon.
\newblock How to sample approval elections?
\newblock In {\em Proceedings of the Thirty-First International Joint Conference on Artificial Intelligence, {IJCAI}}, pages 496--502. ijcai.org, 2022.

\bibitem[Sko21]{skowron2021proportionality}
Piotr Skowron.
\newblock Proportionality degree of multiwinner rules.
\newblock In {\em Proceedings of the 22nd {ACM} Conference on Economics and Computation, {EC}}, pages 820--840. {ACM}, 2021.

\bibitem[SS69]{shapley1969market}
Lloyd~S Shapley and Martin Shubik.
\newblock On market games.
\newblock {\em Journal of Economic Theory}, 1(1):9--25, 1969.

\bibitem[TZZ25]{TaoZZ24}
Biaoshuai Tao, Chengkai Zhang, and Houyu Zhou.
\newblock The degree of (extended) justified representation and its optimization.
\newblock In {\em Proceedings of the 24th International Conference on Autonomous Agents and Multiagent Systems, {AAMAS}}, pages 2024--2032, 2025.

\bibitem[Xia25]{Xia2025linear}
Lirong Xia.
\newblock A linear theory of multi-winner voting.
\newblock {\em CoRR}, abs/2503.03082, 2025.
\newblock URL: \url{https://doi.org/10.48550/arXiv.2503.03082}.

\end{thebibliography}

\newpage
\appendix
\section{Complexity Basics}
\label{sec:complexitybasics}
In this section, we review some basic concepts in complexity theory and define the complexity classes $\Sigma_2^p$ and $\Theta_2^p$. Readers familiar with these can skip this section.

Let $\Sigma$ be an \emph{alphabet-set}, which is typically set to binary $\Sigma=\{0,1\}$, and let $\Sigma^\ast$ be the set of strings of arbitrary length over the alphabet-set $\Sigma$. 
A \emph{language}, or a \emph{decision problem}, is a subset $L$ of $\Sigma^\ast$.
A yes instance of $L$ is a string $x$ with $x\in L$, and a no instance of $L$ is a string $x$ with $x\notin L$.
For example, for \SJRE, yes instances are strings that encode elections where SJR committees exist.
The \emph{complement} of a language $L$, denoted by $\bar{L}$, is defined by $\bar{L}=\{x\mid x\notin L\}$.
Given two languages $L$ and $L'$, we say that $L$ \emph{Karp-reduces to} $L'$, denoted by $L\leq_KL'$, if there is a polynomial time Turing machine that takes a string $x$ as input and outputs another string $x'$ such that $x\in L$ if and only if $x'\in L'$.
All reductions in this paper are Karp reductions, and we will simply say $L$ reduces to $L'$ from now on.
A \emph{complexity class} $X$ is a set of languages.
A language $L$ is $X$-hard if every language in $X$ reduces to $L$.
A language $L$ is $X$-complete if $L\in X$ and $L$ is $X$-hard.
From this definition of hardness, it is straightforward to see that a $Y$-hard language is also $X$-hard if $X\subseteq Y$.

The complexity class $\NP$ is the class of languages that can be decided by some nondeterministic polynomial-time Turing machine $M$.
Specifically, a language $L\in\NP$ if and only if there exists a polynomial-time Turing machine $M$ such that for any input $x$, we have $x\in L$ if and only if there exists a polynomial-length certificate $u$ such that $M(x,u)=1$.
To define the complexity classes $\Theta_2^p$ and $\Sigma_2^p$, we need the notion of \emph{NP oracles}.
Informally, a Turing machine $M$ with an NP oracle, denoted by $M^\NP$, is a Turing machine that can query whether $x\in L$ and get the answer (yes or no) in one step for any language $L\in\NP$.
We refer the readers to standard complexity textbooks for the precise definition of a Turing machine with an NP oracle.

\begin{definition}[$\Sigma_2^p$]\label{def:sigmatwo}
%The complexity class $\Sigma_2^p$ is the second level of the polynomial-time hierarchy.    
A language $L \subseteq \{0,1\}^*$ is in $\Sigma_2^p$ if there exist polynomials $p(\cdot)$, $q(\cdot)$ and a Turing machine $M(\cdot,\cdot,\cdot)$ such that $x\in L$ if and only if
\begin{itemize}
        \item there exists $y\in\{0,1\}^{p(|x|)}$ such that, for all $z \in \{0,1\}^{p(|x|)}$, $M$ accepts the input $(x,y,z)$.
\end{itemize}
%\[
%x \in L \if and only if \exists y \in \{0,1\}^{p(|x|)} \ \forall z \in \{0,1\}^{p(|x|)} \ R(x,y,z).
%\]
In addition, $M$ terminates in at most $q(|x|)$ steps.
\end{definition}

It is a well-known fact that $\Sigma_2^p$ can be equivalently defined as the set of languages in $\NP$ given an $\NP$ oracle.
That is, $\Sigma_2^p$ is the set of languages that can be \emph{verified} in polynomial time given an $\NP$ oracle.
We denote $\Sigma_2^p=\NP^\NP$.

With an $\NP$ oracle, if we consider languages that can be \emph{decided} in polynomial time instead, we get the complexity class $\Delta_2^p$, and we write $\Delta_2^p=\Poly^\NP$.
We will skip the formal definition of $\Delta_2^p$ as it is not further discussed in this paper.

In the definition of $\Delta_2^p$, we do not make any restrictions on the \emph{number} of queries to the NP oracle.
Of course, given that an NP query takes one unit of time and $\Delta_2^p$ requires a decision in polynomial time, we can make at most a polynomial number of queries.
If we are only allowed \emph{logarithmic} numbers of queries, we get the complexity class $\Theta_2^p$.

\begin{definition}[$\Theta_2^p$]\label{def:thetatwo}
    The class $\Theta_2^p$ consists of all languages $L$ for which there exists an oracle Turing machine $M^\NP$, a constant $c>0$, and a polynomial $p(\cdot)$
    such that for every input $x$,
    \begin{enumerate}
        \item $M^\NP(x)$ runs in time at most $p(|x|)$,
        \item $M^\NP$ makes at most $c\log |x|$ queries to the $\NP$ oracle, and
        \item $x\in L$ if and only if $M^\NP$ accepts $x$.
    \end{enumerate}
\end{definition}

Notice that the allowed $O(\log|x|)$ queries in the definition above can be \emph{adaptive}.
If we are only allowed \emph{non-adaptive} queries (i.e., all the queries must be made at the beginning, and the content of one query cannot depend on the result of another query), it is proved independently by Hemachandra~\cite{hemachandra1987strong} and Buss and Hay~\cite{buss1991truth} that $\Theta_2^p$ can be equivalently defined as the set of languages that can be decided with a polynomial number of non-adaptive queries.
For this reason, we denote $\Theta_2^p \;=\; \Poly^{\NP[\log n]}
        \;=\; \Poly^{\NP}_{\parallel}$.
It is obvious that a logarithmic number of adaptive queries can be simulated by a polynomial number of non-adaptive queries by exhausting all the paths in the decision tree, so $\Poly^{\NP[\log n]}\subseteq \Poly^{\NP}_{\parallel}$.
The other direction of containment is much less trivial.

It is easy to see the following relations among the complexity classes discussed above:
\[
\Poly\subseteq\NP\subseteq\Theta_2^p\subseteq\Delta_2^p\subseteq\Sigma_2^p,
\]
and it is widely believed that all the containments above are proper.

\section{Hardness of Committee Verification}
\label{sec:coNPhardness}

In this section, we show that the committee verification problem with both SJR and AJR is coNP-complete.
The proofs for the following two theorems use similar arguments as in Aziz et al.~\cite{aziz2018complexity}.

\begin{definition}[\textsc{Balanced Biclique}]
    Given a bipartite graph $G=(L,R,E)$ and an integer $b\le \min\{|L|,|R|\}$, decide whether there exist subsets $L'\subseteq L$ and $R'\subseteq R$ with $|L'|=|R'|=b$ such that every edge between $L'$ and $R'$ is present, i.e.,
    \[
        L'\times R'\subseteq E.
    \]
\end{definition}
The \textsc{Balanced Biclique} problem is NP-complete~\cite{johnson1981np}, and remains NP-complete when restricted to instances with $b\ge 3$: given any instance, add three universal vertices to each side and increase $b$ by three.

\begin{theorem}
    Given an instance $\mathbb{E}=(N,C,\mathcal{A},k)$ and a committee $W\subseteq C$ with $|W|=k$, it is coNP-complete to decide if $W$ satisfies SJR. 
\end{theorem}
\begin{proof}
    Membership in coNP is immediate. Indeed, if $W$ does not satisfy SJR, then a certificate consists of an integer $t$ and a voter set $N'\subseteq N$. In polynomial time one can check that $N'$ is $t$-cohesive and that some voter $i\in N'$ has $S_W(i)<t$.

    For hardness, we prove that the complement of SJR verification is NP-hard. We reduce from \textsc{Balanced Biclique}. By the preceding observation, we may assume that the input is $(G=(L,R,E),b)$ with $L=\{v_1,\ldots,v_s\}$ and $s\ge b\ge 3$.

    We construct an election as follows. Let $C_L,C_N,C_R,C_D$ be pairwise disjoint candidate sets with
    \[
        |C_L|=|C_N|=b-1,\qquad C_R=R,\qquad |C_D|=sb-3s+2b-2.
    \]
    The voter set is the disjoint union $N_L\cup N_0\cup N_D$, where $N_L=\{x_1,\ldots,x_s\}$, $|N_0|=sb$, and $|N_D|=|C_D|$. For each $x_i\in N_L$, set
    \[
        A_{x_i}=C_L\cup\{r\in R:(v_i,r)\in E\}.
    \]
    Each voter in $N_0$ approves exactly $C_N\cup C_R$. Finally, the candidates in $C_D$ are matched one-to-one with the voters in $N_D$: each voter in $N_D$ approves exactly one candidate of $C_D$, and each candidate of $C_D$ is approved by exactly one voter of $N_D$. Set the committee size to $k=2b-2$ and let
    \[
        W=C_L\cup C_N.
    \]
    The number of voters is
    \[
        n=s+sb+(sb-3s+2b-2)=2(s+1)(b-1),
    \]
    and hence the Hare quota is $n/k=s+1$.

    Suppose first that $G$ contains a balanced biclique $(L',R')$ with $|L'|=|R'|=b$. Let
    \[
        N^*=N_0\cup\{x_i:v_i\in L'\}.
    \]
    Then $|N^*|=b(s+1)=b\,n/k$, and all voters in $N^*$ approve every candidate in $R'$. Thus $N^*$ is $b$-cohesive. However, every voter in $N^*$ has satisfaction exactly $b-1$ under $W$: voters in $N_0$ are represented only by $C_N$, and voters corresponding to $L'$ are represented only by $C_L$. Therefore $W$ violates SJR.

    Conversely, suppose that $W$ violates SJR. Then there exist an integer $t$ and a $t$-cohesive group $N^*$ such that some voter in $N^*$ has satisfaction less than $t$. No voter in $N_D$ can belong to any cohesive group: such a voter shares an approved candidate with no other voter, while every cohesive group has size at least $n/k=s+1>1$. Hence $N^*\subseteq N_L\cup N_0$.

    Every voter in $N_L\cup N_0$ has satisfaction exactly $b-1$ under $W$. Therefore $t\le b-1$ cannot yield an SJR violation. On the other hand, if $t\ge b+1$, then
    \[
        |N^*|\ge t(s+1)\ge (b+1)(s+1)>s(b+1)=|N_L\cup N_0|,
    \]
    which is impossible. Thus $t=b$.

    Since $N^*$ is $b$-cohesive, $|N^*|\ge b(s+1)$. As $|N_0|=sb$, the group $N^*$ contains at least $b$ voters from $N_L$. It also contains at least one voter from $N_0$, because $|N_L|=s<b(s+1)$. Consequently, every commonly approved candidate of $N^*$ lies in $C_R$. Since $N^*$ is $b$-cohesive, there are at least $b$ candidates in $C_R$ approved by all voters in $N^*$. Choosing any $b$ voters from $N^*\cap N_L$ and any $b$ such common candidates gives a $b$-by-$b$ biclique in $G$.

    We have shown that $G$ has a balanced biclique of size $b$ if and only if the constructed committee $W$ fails SJR. Thus the complement of SJR verification is NP-hard, and SJR verification is coNP-hard. Together with membership in coNP, this proves coNP-completeness.
\end{proof}

\begin{theorem}
    Given an instance $\mathbb{E}=(N,C,\mathcal{A},k)$ and a committee $W\subseteq C$ with $|W|=k$, it is coNP-complete to decide if $W$ satisfies AJR. 
\end{theorem}
\begin{proof}
    Membership in coNP is again immediate: a certificate for non-AJR consists of an integer $t$ and a voter set $N'\subseteq N$ such that $N'$ is $t$-cohesive and $S_W(N')/|N'|<t$, both of which can be checked in polynomial time.

    For hardness, use the same construction from an instance $(G=(L,R,E),b)$ of \textsc{Balanced Biclique} as in the proof of the preceding theorem. We show that $G$ has a balanced biclique of size $b$ if and only if the constructed committee $W=C_L\cup C_N$ fails AJR.

    If $G$ has a $b$-by-$b$ biclique $(L',R')$, then the voter set
    \[
        N^*=N_0\cup\{x_i:v_i\in L'\}
    \]
    is $b$-cohesive: it has size $b(s+1)=b\,n/k$, and all its voters approve all candidates in $R'$. Every voter in $N^*$ has satisfaction exactly $b-1$ under $W$, and therefore
    \[
        \frac{S_W(N^*)}{|N^*|}=b-1<b.
    \]
    Hence $W$ violates AJR.

    Conversely, suppose that $W$ violates AJR. Then for some integer $t$ there is a $t$-cohesive group $N^*$ with $S_W(N^*)/|N^*|<t$. As in the preceding proof, no voter in $N_D$ can be contained in any cohesive group. Thus $N^*\subseteq N_L\cup N_0$, and every voter in $N^*$ has satisfaction exactly $b-1$ under $W$. Hence $t\le b-1$ is impossible. Also, no $t$-cohesive group exists for $t\ge b+1$, since then the group-size requirement would exceed $|N_L\cup N_0|=s(b+1)$. Therefore $t=b$.

    The same counting argument now applies: $N^*$ contains at least $b$ voters from $N_L$ and at least one voter from $N_0$, so its common approved candidates all lie in $C_R$. Since $N^*$ is $b$-cohesive, at least $b$ candidates in $C_R$ are approved by all voters in $N^*$. Any $b$ voters from $N^*\cap N_L$ together with any $b$ such candidates form a $b$-by-$b$ biclique in $G$.

    Thus the complement of AJR verification is NP-hard, and AJR verification is coNP-hard. Together with membership in coNP, this proves coNP-completeness.
\end{proof}
\end{document}